\documentclass[runningheads]{llncs}
\usepackage[T1]{fontenc}
\usepackage{graphicx}
\usepackage[margin=1in]{geometry}
\usepackage{amsmath}
\usepackage{amssymb}
\usepackage{comment}
\usepackage{alltt}
\usepackage{natbib}
\usepackage[subtle]{savetrees}
\usepackage{enumitem}
\usepackage{subcaption}
\usepackage{mathtools}
\usepackage{tikz}
\usepackage{bm}
 \usepackage{url}
\usepackage[inkscapelatex=false]{svg}
\usetikzlibrary{calc}
\usetikzlibrary{decorations.pathreplacing}

\usepackage{apptools}
\AtAppendix{\counterwithin{lemma}{section}}
\AtAppendix{\counterwithin{definition}{section}}
\AtAppendix{\counterwithin{theorem}{section}}
\AtAppendix{\counterwithin{figure}{section}}
\AtAppendix{\counterwithin{corollary}{section}}
\begin{document}
\title{Incorporating indel channels into \\ average-case analysis of seed-chain-extend\thanks{This research was funded in part by the National Science Foundation (BIO/DBI grant 2531433)}}
\author{Spencer Gibson\inst{1}\orcidID{0000-0001-8849-6925} \and
Yun William Yu\inst{2}\orcidID{0000-0002-8275-9576}}
\titlerunning{Average-case analysis of seed-chain-extend with indels}
\authorrunning{S. Gibson and Y.W. Yu}
%
\institute{Department of Biomedical Engineering, Carnegie Mellon University, Pittsburgh, PA, USA \and
Ray and Stephanie Lane Computational Biology Department, Carnegie Mellon University, Pittsburgh, PA, USA \\
\email{ywyu@cmu.edu}}
\maketitle              
\begin{abstract}
    Given a sequence $s_1$ of $n$ letters drawn i.i.d. from an alphabet of size $\sigma$ and a mutated substring $s_2$ of length $m < n$, we often want to recover the mutation history that generated $s_2$ from $s_1$. Modern sequence aligners are widely used for this task, and many employ the seed-chain-extend heuristic with $k$-mer seeds. Previously, Shaw and Yu showed that optimal linear-gap cost chaining can produce a chain with $1 - O\left(\frac{1}{\sqrt{m}}\right)$ recoverability, the proportion of the mutation history that is recovered, in $O\left(mn^{2.43\theta} \log n\right)$ expected time, where $\theta < 0.206$ is the mutation rate under a substitution-only channel and $s_1$ is assumed to be uniformly random. However, a gap remains between theory and practice, since real genomic data includes insertions and deletions (indels), and yet seed-chain-extend remains effective. In this paper, we generalize those prior results by introducing mathematical machinery to deal with the two new obstacles introduced by indel channels: the dependence of neighboring anchors and the presence of anchors that are only partially correct. We are thus able     
    to prove that the expected recoverability of an optimal chain is $\ge 1 - O\Bigl(\frac{1}{\sqrt{m}}\Bigr)$ and the expected runtime is $O(mn^{3.15 \cdot \theta_T}\log n)$, when the total mutation rate given by the sum of the substitution, insertion, and deletion mutation rates ($\theta_T = \theta_i + \theta_d + \theta_s$) is less than $0.159$.

\vspace{1em}
\textbf{Code availability}: The code used in the experimental results section, including all simulation software, analysis code, and produced datasets, can be found at \url{github.com/Lazarus42/seed_chainer_indels}.

\end{abstract}

\newpage
\setcounter{page}{1}
\section{Introduction}
String alignment---i.e., determining the best way to match positions of two similar strings $s_1$ and $s_2$ under some cost function---has always been one of the central primitives in computational biology, essential for downstream biological analyses like comparing relatedness of genomes or mapping sequenced reads \cite{britten2002divergence, koonin2000impact, nurk2022complete, siren2021pangenomics}. It is closely related to the edit distance problem \cite{berger2020levenshtein, Chvátal_Sankoff_1975, kiwi2005expected, reinert2000probabilistic, ukkonen1983approximate}, as the choice of matching positions in the alignment corresponds to a series of insertions, deletions, and substitutions to transform $s_2$ into $s_1$.
To this end, Needleman-Wunsch \cite{needleman1970general} and Smith-Waterman \cite{smith1981identification} gave dynamic programming exact solutions to the global and local alignment problems in quadratic $O(mn)$ time, where $|s_1|=n$ and $|s_2|=m$. Although more efficient algorithms exist, e.g., the Four Russians Method has a time complexity of $O(mn/\log(n))$ \cite{four_russians}, it turns out that in the worst case, we cannot do polynomially better---Backurs and Indyk showed in 2015 that ``edit distance cannot be computed in strongly subquadratic time (unless SETH is false)'' \cite{backurs2015edit}.

Of course, there's more to the story.
BLAST \cite{altschul1990basic}, one of the most highly cited papers of all time \cite{van2014top}, gives a linear-time heuristic for local alignment, at the cost of optimality, and is still to this day one of the primary workhorses of bioinformatics, whereas Smith-Waterman has been relegated to being just a subroutine within heuristic software. More broadly, there are many exact sequence alignment/edit distance algorithms that have subquadratic time complexity under certain conditions—Ukkonen's method \cite{ukkonen1985algorithms} for edit distance runs in $O(s \cdot \min(m,n))$ in the worse case, where $s$ is the edit distance between the two strings, and Myer's algorithm \cite{myers1986nd} has $O(m + n + d^2)$ average-case time complexity where $d$ is the minimum edit script between the two strings. Other than exact algorithms, there are numerous heuristics that optimize sequence alignment for specific tradeoffs \cite{chaisson2012mapping, groot2022exact, ivanov2022fast, langmead2012fast}. 

One heuristic of particular interest is ``seed-chain-extend,'' which is used in modern software such as Minimap 2 \cite{li2018minimap2}. The seed-chain-extend heuristic has three stages. In \textit{seeding}, $k$-mer ``seeds'' are selected on both $s_1$ and $s_2$, and shared k-mers are marked as ``anchors'' between the two strings. Afterwards, concordant anchors are \textit{chained} together to form the skeleton of an alignment. Lastly, the space between anchors is filled in using standard worst-case quadratic-time dynamic programming in a process known as \textit{extension}.
Seed-chain-extend empirically showcases near quasilinear runtime on the similar genomic strings it is typically applied to, but is not guaranteed to find optimal alignments.

For a long time, bioinformaticians have contented ourselves to this gap between theory and practice.
In the last five years, though, theoreticians have made several new breakthroughs by defining a generative model of string evolution and revisiting average-case analysis. Analysis of string algorithms \cite{bukh2020length, lember2009lcs}, particularly average-case analysis \cite{szpankowski2011average}, historically made extensive use of generating functions; however, more recent approaches instead used tail-bounds to bound bad events, more akin to the analysis of some randomized probabilistic sketches \cite{yu2022hyperminhash}. Ganesh and Sy's 2020 breakthrough was to show that under their random mutation model, a modified dynamic programming algorithm will compute edit distance in $O(n \log n)$ time between a random string $s_1$ and a mutated string $s_2$ of near equal length \cite{ganesh2020near} with high probability.
However, their analysis worked because of concordance of their DP algorithm with their mutation model, so extending it to more practical but sophisticated heuristics like seed-chain-extend was nontrivial.

In 2023, our research group made substantial progress on proving similar results for seed-chain-extend \cite{shaw2023proving}.
Unfortunately, the machinery and techniques we introduced in that paper were insufficiently powerful to address insertions and deletions, so we had to restrict our theoretical results to a substitution-only mutation model.
Indels tend to complicate analyses and dependence structures, so most bioinformatics theoreticians either avoid directly working with them \cite{edgar2021syncmers,ondov2016mash,shaw2022theory,yu2015quality}, or adjust their algorithm and model to directly capture them \cite{ganesh2020near}.
We were able to run empirical benchmarks with indels that closely tracked our substitution-only theory (including accurate predictions of exponents), so we believed without proof the theorems were also correct for indels \cite{shaw2023proving}.
In this sequel, we make progress on that remaining gap between theory and practice, generalizing our prior machinery and techniques to handle indels.

\section{Problem Introduction and Strategy Overview}

\subsection{The seed-chain-extend algorithm}
Given strings $S$, $S'$, we can ask about the likely sequence of edits transforming $S$ into $S'$. Seed-chain-extend answers this question with three steps: seeding, chaining, and extension. 
Seeding is the process by which k-mers (``seeds") are selected on both $S$ and $S'$ and matching seeds are called anchors. Formally, an anchor of length $k$ occurs at $(i,j)$, i.e., starting at position $i$ in $S$ and position $j$ in $S'$, if $S[i:i+k-1] = S'[j:j+k-1]$ (indexing is right-inclusive). 
Second is chaining: a chain $\mathcal{C} = ((i_1, j_1), \dots, (i_u, j_u))$ is a sequence of anchors where $i_\ell < i_{\ell + 1}$ and $j_{\ell} \le j_{\ell + 1}$ for all $1 \le \ell < u$. The chaining score used to determine the optimal chain in this paper is the linear gap cost. For the given chain $\mathcal{C}$, this would be $u - \xi (i_u - i_1 + j_u - j_1)$, or, equivalently, $u - \sum_{\ell=2}^u \xi (i_{\ell + 1} - i_{\ell} + j_{\ell + 1} - j_{\ell})$.
Given the optimal chain $\mathcal{C} = ((i_1, j_1), \dots, (i_u, j_u))$, we use classical dynamic programming to extend between consecutive anchors based on any alignment score. In this work, we assume that the reference string $S$ is already seeded, while the query string $S'$ is not.

\subsection{Challenges to analysis of seed-chain-extend}
Several difficulties arise in proving average-case results.
First, chaining and extension have different optimization objectives.
When considering overall performance of the heuristic, it is possible for failures to happen at any stage.
By failure, we mean anything that leads to bad downstream events---not finding the correct string alignment or taking too long to find that alignment.
A bad chain can result from a failure in chaining, or just a bad selection of anchors in seeding. Similarly, extension failure can be a result of bad chaining, or because the extension procedure does not itself find the right alignment, despite the chaining being ``good''.
Also note that the alignment problem under a mutation model diverges somewhat from the edit distance problem.
The best scoring alignment corresponds to some edit distance between the strings, but arguably the ``correct'' alignment (at least from a biological perspective) is the one that reflects the mutations that happened to transform $s_1$ into $s_2$.

To resolve these difficulties, in the prequel \cite{shaw2023proving}, we introduced the concept of ``recoverability'', which decoupled chaining accuracy from extension accuracy. Extension is only performed in the gaps between anchors on the chain, so roughly speaking, recoverability measures how much of the correct alignment can possibly be recovered given a chain.
By structuring our problem thus, we can focus on just the seeding and chaining---i.e. how good is the chain as a starting point for the extension phase.
Importantly, recoverability of a chain is a theoretical measure of the goodness of the chain, as opposed to the optimization criterion used to generate the chain, such as linear-gap cost chaining.

\subsection{Difficulty from indels}
Unfortunately, our definition of recoverability in the prequel relied on the substitution-only error model, where the correct alignment of a mutated substring $s_2$ to $s_1$ is a diagonal in the alignment matrix (the ``homologous diagonal'').
Recoverability was defined as the proportion of the homologous diagonal covered by anchors on the optimal chain and the dynamic programming extension blocks between anchors.
With indels, the correct alignment is no longer a straight diagonal.
In this sequel, we must generalize the homologous diagonal to a ``homologous path''.

In the same vein, indels also mess up the matching of indices between $s_1$ and $s_2$. This is not only notationally very tricky to reconcile, but also make it hard to discuss the dependence structure of positions in anchors, which is necessary for the limited-dependence Chernoff bounds we used in the prequel.

Finally, having indels allows for ``no operations'' (``no-ops'') -- sequences of mutations that preserve the string -- e.g. deletion and insertion of the same letter.
This creates a problem with defining recoverability of a homologous path---the alignment specified by the homologous path will have a kink that cannot be found via $k$-mer matching.
The ``correct'' alignment will not be found by any reasonable alignment method. This is not a problem for recoverability in extension regions, as it still could be an alignment produced by the extension block, but it cannot be found as an anchor, which only gets exact matches.
In addition to exact no-ops, where the correct and lowest distance alignments no longer match, indels can also create alternate equivalent edit distance paths, of which only one can be the path history, such as if an insertion duplicates existing letters
\vspace{-0.5em}\begin{alltt}
         Exact no-op                                  Partial no-op
     AC-GT         ACGT                          ACG--T            A--CGT
     ||  |   vs.   ||||                          |||  |     vs.    |  |||
     ACG-T         ACGT                          ACGCGT            ACGCGT
\end{alltt}
\vspace{-0.5em}
It turns out that this problem cannot be resolved satisfactorily with the prequel's definition of recoverability, so we will have to introduce a more sophisticated definition of recoverability.

\subsection{Proof structure and motivation}
The basic intuition behind the prequel \cite{shaw2023proving} is that given reasonably low mutation rates, the optimal chain under linear-gap-cost chaining will be close to correct in the sense that most of the anchors will lie on the homologous diagonal. A gap between anchors can either be homologous if the anchors flanking it are both on the homologous diagonal, or non-homologous if at least one of the two anchors is off the homologous diagonal. Non-homologous gaps can lead to ``breaks'', which are regions of the string where extension through the gaps does not cover the homologous diagonal, leading to a decrease in recoverability. However, with high probability, each break has size $<\sqrt{m}$ and the number of breaks is small, so the recoverability will be high. Additionally, the runtime can be bounded by extension time through all the homologous and non-homologous gaps, which are small in an optimal chain.

Roughly speaking, the intuitive reason the above strategy works is that substitutions are much more likely to break anchors than they are to create spurious anchors. So long as we remain in the regime where sufficiently many anchors can still be found, the optimal chain is close to the correct chain, and the recoverability will be high.
When indel channels are added to the mix, the above logic still basically holds. The major difference we handle is a new class of anchors, termed ``clipping'' anchors (Fig. \ref{fig:combined-recoverability-matchgraph}A), which lie partially on the homologous path. Clipping anchors contribute to recoverability in a similar way as homologous anchors: extending through gaps flanked by clipping anchors still contains the bulk majority of the path in the gap. However, they behave like spurious anchors in that the anchors themselves may not recover points on the homologous path. 

By generalizing all the theorems in the prequel \cite{shaw2023proving}, we conclude that the expected (prequel) recoverability of an optimal chain is $\ge 1 - O\bigl((\log n)^2\,(1-\theta)^k\bigr)$. The reader will find proofs for mathematically involved lemmas and theorems in Appendix \ref{sec:old-recoverability}. Notice that this is asymptotically worse than the $1- O(\frac{1}{\sqrt{m}})$ recoverability in the substitution-only setting, and does not match empirics, necessitating the recoverability generalization we do here.

\subsection{Preliminaries}\label{sec:preliminaries}

In this section, we present three important concepts for our analysis: the mutation model, which defines how the string $S'$ is given from $S$, the homologous path, which represents the edit history transforming $S$ into $S'$, and the recoverability of a chain given the edit history. Informally, the mutation model consists of passing a substring of $S$ through a channel. Independently, at each letter of this substring of $S$, there can be a substitution, a deletion, or a random string inserted to the left:

\begin{definition}
    [Mutation model] Let $S = x_1 x_2 \cdots x_{n}$, $n = |S|$, be a string where each letter \(x_i\) is sampled i.i.d.\ from an alphabet of size \(\sigma\).  The mutated substring $S'$, $m = |S'|$, is obtained by passing the substring $S[p+1:p+m']$ through a mutation channel where for each position $p + j$, $1 \le j \le m'$, the input symbol \(S[p+j]\) undergoes the following mutations independently:
    \vspace{-0.5em}
\begin{itemize}
  \item \textbf{Substitution} (probability \ \(\theta_s\)): \(S[p+j]\) is replaced by a different letter, chosen uniformly.
  \item \textbf{Deletion} (probability \ \(\theta_d\)): \(S[p+j]\) is deleted.
  \item \textbf{Insertion} (probability \ \(\theta_i\)): a random string of length $L \sim \mathrm{Geom}(1-\rho_i')$ is inserted at position \(p+j\).
\end{itemize}
\end{definition}

\begin{remark}
    The probability that at position $p + j$, where $p+1 \le p + j \le p+m'$, there is no insertion, deletion, or substitution is $(1 - \theta_d)(1 - \theta_s)(1 - \theta_i) = 1 + (\theta_i \theta_d + \theta_i \theta_s  + \theta_d \theta_s) - (\theta_d + \theta_s + \theta_i) \ge 1 - (\theta_d + \theta_s + \theta_i)$. We will refer to the ``total'' mutation rate as $\theta_T = \theta_d + \theta_s + \theta_i$.
\end{remark}

In the definition above, insertion is to the ``left'' of the position at which it occurs. For example, if $S = ACCTAG$, and the string $TTACA$ is inserted at position $2$, then $S'$ would be $AC(TTACA)CTAG$.

If the mutation model only allows for substitutions (\cite{shaw2023proving}), the edit history is simple: each position in $S'$ is matched with a single corresponding position in $S$. With this, the optimal alignment was defined simply to be the diagonal path matching every position $p+j$ in $S$ to $j$ in $S'$, i.e., the set of points $\{(p+j, j) \mid 1 \le j \le m' \}$. 

We extend the same intuition to the indel case: the optimal alignment should capture the edit history transforming $S$ into $S'$. We call this edit history the \textit{homologous  path}, which generalizes the homologous diagonal from the prequel \cite{shaw2023proving}. One reasonable way to define the edit history is following the canonical alignment from Ganesh and Sy \cite{ganesh2020near}, which is detailed below.

\begin{definition}[Homologous path (Inspired by Ganesh and Sy)]\label{def:canonical-alignment}
Let $S$ be the original length $n$ string with i.i.d. letters sampled from the alphabet of size $\sigma$, and $S'$ the string resulting from passing $S[p+1:p+m']$ through the mutation channel where edits $\mathcal{E}$ are applied.

The homologous path $P_H$ between $S$ and $S'$ is represented as a list that initially contains \((p,0)\). Let $(i,j)$ be the last element of $P_H$. We iteratively append points to $P_H$ based on the mutations that occur at position $i + 1$ in $S$, assuming $p + 1 \le i + 1 \le p + m'$, following $\mathcal{E}$:
\vspace{-0.5em}
\begin{itemize}
  \item If no insertion or deletion occurs at position \(i + 1\), i.e., a substitution occurs or there is no mutation, append the point $(i+1,j+1)$ to $P_H$.

  \item If an insertion of \(I\) letters occurs at position \(i + 1\) and no deletion occurs, append the points
  $(i,j+1),\,\dots,\,(i,j+I),\,(i+1,j+I+1)$ to $P_H$.

  \item If a deletion and no insertion occurs at position \(i + 1\), append the point $(i+1,j)$ to $P_H$.

  \item If both an insertion of \(I\) letters and a deletion occur at position \(i + 1\), append the points
  $(i,j+1),\,\dots,\,(i,j+I),\,(i+1,j+I)$ to $P_H$.
\end{itemize}
\end{definition}

As an example, consider the string $S$ = TACTTCGC, which undergoes the following mutations: T is inserted at position $4$, position $5$ is deleted, the letter at position $6$ is substituted to a T, and the letter at position $7$ is substituted to an A. These edits result in a string $S'$=TACTTTAC. The dependence structure (match graph) between $S$ and $S'$ is shown in Fig.~\ref{fig:combined-recoverability-matchgraph}C, and the homologous path is shown in the Appendix (Fig.~\ref{fig:homologous_path}).

\begin{figure}[tbp]
\centering

\begin{subfigure}[t]{0.65\textwidth}
    \centering
    \includesvg[width=\textwidth]{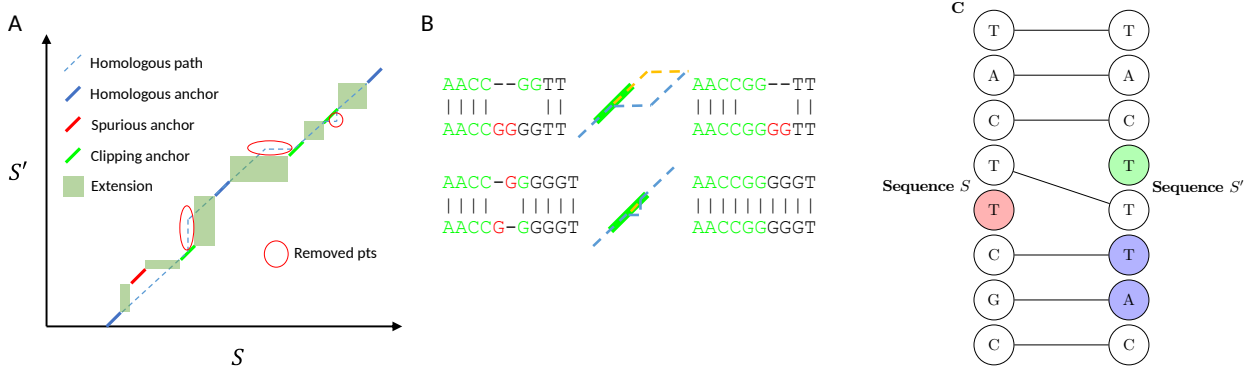}
\end{subfigure}%
\hfill%
\begin{subfigure}[t]{0.3\textwidth}
\centering
\resizebox{\textwidth}{!}{%
\begin{tikzpicture}[scale=0.9, every node/.style={font=\small}]
  \node[font=\bfseries] at (-0.8,-0.5) {\textbf{C}};

  \foreach \i/\base in {1/T,2/A,3/C,4/T,6/C,7/G,8/C} {
    \node[circle,draw,minimum size=8mm] (S\i) at (0,-\i) {\base};
  }
  \node[circle,draw,fill=red!30,minimum size=8mm] (S5) at (0,-5) {T};

  \node[left=10pt] at (0,-4.5) {\textbf{Sequence }$S$};

  \foreach \i/\base in {1/T,2/A,3/C,5/T} {
    \node[circle,draw,minimum size=8mm] (Sprime\i) at (3,-\i) {\base};
  }

  \node[circle,draw,fill=green!30,minimum size=8mm] (Sprime4) at (3,-4) {T};

  \node[circle,draw,fill=blue!30,minimum size=8mm] (Sprime6) at (3,-6) {T};
  \node[circle,draw,fill=blue!30,minimum size=8mm] (Sprime7) at (3,-7) {A};
  \node[circle,draw,minimum size=8mm] (Sprime8) at (3,-8) {C};

  \node[right=10pt] at (3,-4.5) {\textbf{Sequence }$S'$};

  \draw (S1) -- (Sprime1);
  \draw (S2) -- (Sprime2);
  \draw (S3) -- (Sprime3);
  \draw (S4) -- (Sprime5);
  \draw (S6) -- (Sprime6);
  \draw (S7) -- (Sprime7);
  \draw (S8) -- (Sprime8);
\end{tikzpicture}%
} 
\end{subfigure}

\caption{
\textbf{(A-B)} Generalized recoverability. (a) Homologous anchors lie entirely on the path, spurious anchors entirely off, and clipping anchors partially on/off. For the purposes of recoverability, we remove points that share an x- or y-coordinate with a clipped point. (b) Removing points corresponds to allowing alternate just-as-good paths where the clipping anchor is homologous.
\textbf{(C)} The match graph resulting from the mutation process that gives $S'$ from 
$S = \text{TACTTCGC}$, including a deletion (red), insertion (green), substitutions (blue), and matches (clear). Horizontal lines represent corresponding positions between the sequences. Specifically, in $S$, an insertion of the letter T occurs at position $4$, position $5$ is deleted and the characters at positions $6,7,$ and $8$ are mutated.
}
\label{fig:combined-recoverability-matchgraph}
\end{figure}

We end this section by defining constants and bounds on constants that will be used throughout.

\begin{definition}[Key definitions and assumptions]\label{sce_defs}
    The constant $\alpha = -\log_{\sigma}(1-\theta_T)$ represents the expected per-base contribution to the matching length between $S$ and $S'$. We will consistently refer to the lengths $n = |S|$, $m' = |S[p+1:p+m']|$, and $m = |S'|$.
    
    The length of the seeds, $k$, is given by $k = C\log_{\sigma}(n)$ for any $C > \frac{3}{1 - 2\alpha}$. We will write $\log = \log_{\sigma}$ for short. We can write $C\alpha = \frac{3\alpha}{1 - 2\alpha} + \delta$. Since $\sigma \ge 4$ and $\theta_T < 0.159$, we have that $\alpha < \frac{1}{8}$. We can always choose $\delta > 0$ small enough such that $C\alpha < \frac{1}{2}$. 

    Let $c_0 = \max(\frac{1}{2}\ln(\frac{9}{1 + 8\gamma}), \frac{21}{\beta}) \le 30$, where $\beta=\log_{\sigma}(e)$. The gap penalty constant in the seed chain extend linear-gap cost will be $\xi = \frac{1}{n}$, and $g(n) = \frac{50k}{8(1 - \theta_T)^k} \ln(n)$.
    
    The length of the generative region in $S$ is $m' = \Omega(n^{2C\alpha + \epsilon})$, for an arbitrarily small enough $\epsilon > 0$. Note such a choice is possible since $C\alpha < 1/2$ and $n^{2C\alpha + \epsilon} < n$. By our choice of variables, $\sigma^k = n^C$ and $(1-\theta_T)^k = n^{-C\alpha}$, which we will use repeatedly. Note also that $k = C\beta\text{ln}(n)$.

    Recall that the cumulative mutation rate is $\theta_T = \theta_s + \theta_d + \theta_i$. We require that $\theta_T < 0.159$.

\end{definition}

Note that $\frac{3\alpha}{1 - 2\alpha} = \frac{-3\log(1-\theta_T)}{2+\log(1-\theta_T)}$ is convex in $0 \le \theta_T < 0.159$. Then, because $C > \frac{3}{1 - 2\alpha}$, and $C\alpha = \frac{3\alpha}{1 - 2\alpha} + \delta$ for any $\delta > 0$, an upper bound is $C\alpha < 3.15 \cdot \theta_T$, slightly worse than the $2.43\theta$ upper bound in the substitution-only case~\cite{shaw2023proving}.
As a concrete example, if $\theta_T = 0.05$, the divergence between human and chimpanzee genomes, $C\alpha \approx \frac{-3\log(0.95)}{1 + 2\log(0.95)} \le 0.12$, so runtime of seed-chain-extend would be $O(n^{1.12}\log(n))$. If $n = |S| = |S'| \approx 3\times 10^9$, the size of the human genome, $n^{1.12} \approx 40 n$, and the runtime of seed-chain-extend is essentially $O(n\log n)$.

Intuitively, string alignment should recreate the simplest edit history for a given scoring function. Excluding ``no-ops'' from the set of recoverable points ensures that our definition does not penalize such simplicity. For example, if $(i,j)$ is a point in $P_H$ and $L > 0$ is maximal such that $S[i+1:i+L] = S'[j+1:j+L]$ and $(i+\ell, j+\ell) \notin P_H$ for $1 \le \ell \le L$, then any reasonable alignment method will align $(i+\ell, j+\ell)$. Since the edits causing $P_H$ to not include these points leaves the substrings of $S$ and $S'$ unchanged, all points in the offending region are removed when determining recoverability. 

\begin{definition}[Defining non-recoverable regions]
    For a point $(i,j) \in P_H$, define the right non-recoverable region $NR(i,j)$ and the left non-recoverable region $NL(i,j)$ as follows:

    First we will define the lengths $r(i,j)$ and $l(i,j)$ of the non-recoverable regions.

\noindent    $r(i,j) = \max\{t \ge 0 \mid S[i+1:i+t] = S'[j+1:j+t]$ and $(i+\ell, j+\ell) \notin P_H$ for $1 \le \ell \le t\}$ if such a $t$ exists, and $0$ otherwise.

\noindent     $l(i,j) = \max\{t \ge 0 \mid S[i-1:i-t] = S'[j-1:j-t]$ and $ (i-\ell, j-\ell) \notin P_H$ for $1 \le \ell \le t\}$ if such a $t$ exists, and $0$ otherwise.

    Then the right non-recoverable region is $NR(i,j) = \{(x,y) \in P_H \mid i< x \le i + r(i,j) \lor j < y \le j + r(i,j)\}$.

    The left non-recoverable region is $NL(i,j) = \{(x,y) \in P_H \mid i \ge x > i - l(i,j) \lor j \ge y > j - l(i,j)\}$.
\end{definition}

Having defined the set of non-recoverable points, our generalized recoverability can be defined by removing all non-recoverable points from the homologous path. The recoverability of a chain, $\mathcal{C} = ((i_1,j_1),\dots,(i_u,j_u))$, loosely speaking, is the fraction of the homologous path $P_H$, that lies in the extension of anchors in the chain and in gap extensions. We define recoverability to account for all \textit{possible} alignments given by the chain, which means that we count any portion of the homologous path in an extension as ``recovered'' since, in theory, it could be recovered by some extension algorithm. Recoverability will be the fraction of recoverable points that are, in fact, recovered by the chain. We formalize this intuition below. 

\begin{definition}[Generalized recoverability]
    Given a chain $\mathcal{C} = ((i_1,j_1),\dots,(i_u,j_u))$, we define the union of all possible alignments for the chain $\mathcal{C}$, $Align(\mathcal{C})$, as:
    $$Align(\mathcal{C}) = \bigcup_{\ell=1}^u \{ (i_\ell,j_\ell),\dots, (i_\ell +k-1, j_\ell + k - 1)  \} \,\ \cup \,\ \bigcup_{\ell=1}^{u-1} Ext(\ell).$$ Where $Ext(\ell) = \{i_\ell + k - 1,\dots,i_{\ell+1}\} \times \{j_\ell + k - 1,\dots,j_{\ell+1}\}$. If $i_\ell + k - 1> i_{\ell + 1}$ or $j_\ell + k - 1 > j_{\ell + 1}$, then $Ext(\ell) = \emptyset$. 
    
    Denote the set of non-recoverable points as $U = \bigcup_{(i,j) \in P_H}(NR(i,j)\cup NL(i,j))$. The recoverability of the chain, $R(\mathcal{C})$, is defined to be:
    $$R(\mathcal{C}) = \frac{|Align(\mathcal{C}) \cap P_H\setminus U|}{|P_H \setminus U|}$$
    
\end{definition}

Lastly, we introduce the three types of anchors that are relevant to our analysis: homologous, clipping, and spurious anchors (Fig.~\ref{fig:combined-recoverability-matchgraph}A). Visually, anchors always appear as diagonals of length $k$ in the alignment matrix. A homologous anchor lies entirely on the homologous path. Unlike the substitution-only case, under indel channels, anchors can touch the homologous path at a number of points without lying entirely on it-- these are called clipping anchors. Spurious anchors remain the same as in the prequel \cite{shaw2023proving}: anchors that lie entirely off the homologous path. We now formally define each anchor type.

\begin{definition}[Defining anchor types]\label{anchor-types}
    For notational ease, let $A = \{(i+t, j+t) \mid 0 \le t \le k-1\} $ and $B = \{(x,y) \in P_H \mid i \le x \le i+k-1 \ \land\ j \le y \le j+k-1\}$. If there exists an anchor at $(i,j)$, then $A$ is the set of points belonging to that anchor, and $B$ is the set of points on the homologous path between (inclusive) $(i,j)$ and $(i+k-1,j+k-1)$. The anchor at $(i,j)$ is \textbf{homologous} if $A = B$, \textbf{spurious} if $A \cap B = \emptyset$, and \textbf{clipping} otherwise.

\end{definition}

\subsection{Anchor-count concentration bounds}\label{sec:anchor-concentration}

The general strategy of our proof is to show that with high probability there are no spurious anchors and that there are not large uncovered gaps in the homologous path before the first anchor or after the last anchor in any optimal chain. In Appendix \ref{sec:uncovered-gap-bounds}, we analyzed the dependency structure between $S$ and $S'$, which enables us to derive concentration inequalities and tail bounds on the number of spurious anchors $N_S$.

In line with the general proof strategy, this section establishes two main results. The first is an expansion-contraction bound on k-mer blocks in $S$ and $S'$: k-mers in $S[p+1:p+m']$ cannot correspond to more than $O(k)$ points in $S'$, and k-mers in $S'$ cannot correspond to more than $O(k)$ points in $S$. When the expansion-contraction bound holds, which it will with high probability, the homologous path will be relatively well-behaved. The anchor independence lemmas derived in Appendix \ref{sec:uncovered-gap-bounds} will then allow us to bound the variance of the number of spurious anchors, $N_S$, and to use this bound to show that, with high probability, no spurious anchors occur \textit{at all}. 

We proceed with the first result: any $k$-mer in $S[p+1:p+m']$ cannot expand by more than $O(k)$ and blocks of length $\ell = \frac{21}{\beta}$ in $S[p+1:p+m']$ cannot contract to size less than $\frac{(1-\theta_d)\ell}{2}$ with high probability in the generation of $S'$. The proof of this result relies on two key lemmas in Appendix \ref{convergence_bounds}: the bounded expansion and bounded contraction lemmas. The idea behind their proofs is to show that the probability of expansion or contraction, for each relevant block, is $\le \frac{1}{n^2}$ through Chernoff bounds and moment-generating function inequalities. Applying a union bound over all blocks, of which there are at most $n$, gives a probability of $\le \frac{1}{n}$ that any bad event occurs.

\begin{lemma}[Defining the Expansion-Contraction ($\bm{EC}$) space]\label{lemma:EC_space}
    Let $t_0 = \frac{1}{2}\ln(\frac{9}{1 + 8\gamma})$ and $\ell = \frac{21k}{\beta}$. With probability $ \ge 1 - \frac{2}{n}$, no $k$-mer in $S[p+1:p+m']$ has more than $\frac{1}{t_0}(\frac{2}{\beta} + 1)k$ inserted base pairs in $S'$ and no $\ell$-block in $S[p+1:p+m']$ contracts to size $\le \frac{(1 - \theta_d)\ell}{2}$.
\end{lemma}
\begin{proof}
    Let $E_1$ be the event that no $k$-mer in $S[p+1:p+m']$ has more than $\frac{1}{t_0}(\frac{2}{\beta} + 1)k$ inserted base pairs in $S'$ and let $E_2$ be the event that no $\ell$-block in $S[p+1:p+m']$ contracts to size $\le \frac{(1 - \theta_d)\ell}{2}$.
    
    By Lemmas \ref{expansion-bound} and \ref{contraction-bound}, we have $\Pr(E_1^c) \le \frac{1}{n}$ and $\Pr(E_2^c) \le \frac{1}{n}$. A simple union bound gives that $\Pr(E_1^c \lor E_2^c) \le \frac{2}{n}$, from which it immediately follows that $\Pr(E_1 \land E_2) \ge 1 - \frac{2}{n}$.
\end{proof}

We will refer to the space where the bounded contraction and expansion lemmas are jointly satisfied as $\bm{EC}$. Under the $EC$ space, we bound the variance of the number of spurious anchors. Using the variance of $N_S$, we then bound its value with high probability.

\begin{lemma}[Working in $\bm{EC}$]\label{var_bound}
    $\mathbb{E}(N_S^2) \le \mathbb{E}(N_S)^2 + 2T_0k^2\frac{mn}{\sigma^k}$ where $T_0 = 2\max((\frac{1}{t_0}(\frac{2}{\beta} + 1) + 3)^2, 4, \frac{21}{\beta})$. Thus, $\mathrm{var}(N_S) \le T_0k^2\frac{mn}{\sigma^k}$.
\end{lemma}

The below lemma uses the conditional bounded variance of $N_S$ from the previous result to bound the number of spurious anchors with high probability. From the lemma below, we will conclude that for large enough $n$, there are no spurious anchors with probability $\ge 1 - \frac{3}{n}$.

\begin{lemma}[Working in $\bm{EC}$]\label{lemma: spurious_anchor_count}
With probability at least $1 - \frac{1}{n}$, the number of spurious anchors is 
$
\le n^{2 - C} + \sqrt{T_0\,}\,C\log(n)\,n^{\frac{3-C}{2}}.
$
Precisely,
\vspace{-1em}
\[
\Pr\Bigl(N_S \ge n^{2 - C} + \sqrt{T_0\,}\,C\log(n)\,n^{\frac{3-C}{2}} \mid EC\Bigr)
\;\le\;\frac{1}{n}
\]
\end{lemma}

\begin{lemma}[Defining the $\bm{F1}$ space]\label{lemma:f1}
    Under $EC$, for large enough $n$, there are no spurious anchors with probability $\ge 1 - \frac{1}{n}$.
\end{lemma}
\begin{proof}
    Note that $C > \frac{3}{1-2\alpha} > 3$. For large enough $n$, we have that $1 > n^{2 - C} + \sqrt{T_0\,}\,C\log(n)\,n^{\frac{3-C}{2}}$. By the previous lemma, $\Pr(N_S = 0) = \Pr(N_S < 1) \ge 1 - \frac{1}{n}$.
\end{proof}

Similar to how the $EC$ space represents the set of events where the expansion-contraction lemma holds, when we work in the $F1$ space, there will be no spurious anchors. $F1$ is also a high probability space: $\Pr(F1) \ge \Pr(F1 \mid EC)\Pr(EC) \ge (1 - \frac{1}{n})(1 - \frac{2}{n}) \ge 1 - \frac{3}{n} \implies \Pr(F1) \ge 1 - \frac{3}{n}$. In the next section, we show that there are no long homologous gaps in the generative region of $S$, and use this fact to bound break lengths with the help of the expansion-contraction lemma. 

\subsection{Bounding homologous gaps}\label{sec:bounding_gaps}

Recall that the general strategy of our proof is to show that with high probability, there are no spurious anchors, and that the vast majority of the path is either covered by (1) a homologous or clipping anchor, or (2) belongs to an extension region between two anchors. In the previous section, we showed that spurious anchors do not occur with high probability. This section is dedicated to the second portion: proving that there exists some homologous anchor, so that the chain is non-empty, and that they are commonly present along the homologous path with high probability.

A homologous gap is defined to be a region $S[p+a,p+b]$, $1 \le a,b \le m'$ in the generative portion of $S$, for which there are no homologous anchors. We begin by bounding the length of any homologous gap by establishing concentration bounds on homologous anchors in this $k$-dependence case. The lemmas in this section make use of the generative region of $S$, which has length $m'$ but the final inequalities are in terms of $|S'| = m$. This is fine because the two lengths are equivalent up to a constant ($m' = cm$ for a constant $c > 0$) while working under the $EC$ space. We will represent fixed constants with variants of $c$; the exact values do not make a difference in the analysis.

We begin with a dependent Chernoff-Hoeffding bound that appears in Yu and Shaw's paper \cite{shaw2023proving}, originally appearing in Janson's paper \cite{janson2004large}. The purpose of this lemma is to obtain an exponential (Chernoff) bound for the number of homologous anchors. The lemma makes use of the local dependence of homologous random variables since overlapping k-mers are dependent but non-overlapping k-mers are independent. 

\begin{lemma}
(Yu and Shaw) Suppose we have
$X = \sum_{a \in A} \text{Bernoulli}_a(q) \quad \text{for some} \quad 0 < q < 1.$
A proper cover of \( A \) is a family of subsets \( \{A_i\}_{i \in I} \) such that all random variables in \( A_i \subset A \) are independent and \( \bigcup_{i \in I} A_i = A \). Let \( \chi(\mathcal{A}) \) be the minimum size of the cover, \( |I| \), over all possible proper covers. Then for \( t \geq 0 \),
\[
\Pr\left( X \leq \mathbb{E}X - t \right) \leq \exp\left( -\frac{8t^2}{25 q|A| \chi(\mathcal{A})} \right).
\]
\end{lemma}

 Let $N_P$ (``P'' for preserved, as these regions have no mutations) be the random variable representing the number of $k$-blocks in $S$ such that there are no mutations in that block. Clearly, each $k$-block in $S$ without mutations is a homologous anchor, i.e., $N_H \ge N_P$. We can apply the previous theorem to establish a tail bound on the number of homologous anchors with a chain of inequalities: first, we obtain a tail bound for $N_P$, and then we bound $N_H$ using that $N_H \ge N_P$.

\begin{lemma}
    For any $0 \le t \le m'((1-\theta_i)(1-\theta_d)(1-\theta_s))^k$, we have that 
    $\Pr(N_H \le m'((1-\theta_i)(1-\theta_d)(1-\theta_s))^k - t) \le \exp(-\frac{8t^2}{25m'k((1-\theta_i)(1-\theta_d)(1-\theta_s))^k})$
\end{lemma}
\begin{proof}
We use the previous theorem with \( q = ((1-\theta_i)(1-\theta_d)(1-\theta_s))^k \), the probability that a k-block has no mutation, as shown in Corollary \ref{cor:kblock_prob}. 
Recall that $E_{i:i+k-1}$ denotes the random variable taking on the value of $1$ if there are no mutations in $S[i:i+k-1]$ and $0$ otherwise. Then each set $A_j= \{E_{j+tk:j+(t+1)k-1} \mid t \ge 0 \land j+tk \le p+m' \}$ contains mutually independent random variables by Corollary \ref{cor:hom_indepdence}. Note $A = \bigcup_{i \in \{1,\dots,m'\}} A_i$, and thus the $A_i$ form a partition of $A$. This implies that $\chi(\mathcal{A}) \leq k$. Applying the previous theorem yields $\Pr(N_P \le m'q - t) \le \exp(-\frac{8t^2}{25m'kq})$. Since $N_H \ge N_P$, it follows that $\Pr(N_H \le m'q - t) \le \Pr(N_P \le m'q - t)$, from which the result follows.
\end{proof}

We can apply the previous lemma to obtain a Chernoff bound for the probability of a gap of length $\ell$ occurring. From this lemma, it will follow that there is no ``large'' homologous gap with probability $\ge 1 - \frac{1}{n}$.

\begin{lemma}
For any interval of length $\ell$ in $S[p+1:p+m']$, the probability that no homologous anchor occurs is 
\[
 \le \exp \left( -\frac{8\ell(1-\theta_T)^k}{25k} \right).
\]
\end{lemma}

\begin{proof}
Each interval of length $\ell$ in $S[p+1:p+m']$ can be viewed as an identically distributed length $\ell$ version of it. Applying the previous lemma with $m' = \ell$ and $t = \ell (1-\theta_T)^k$ gives $\mathrm{Pr}(N_H \le 0) = \mathrm{Pr}(N_H = 0) \le \exp(-\frac{8\ell ((1-\theta_i)(1-\theta_d)(1-\theta_s))^k}{25k}) \le \exp(-\frac{8\ell(1-\theta_T)^k}{25k})$ since $((1-\theta_i)(1-\theta_d)(1-\theta_s))^k \ge (1 - \theta_T)^k$.
\end{proof}

Using the previous lemma, we show that every region of length  $g(n) = \frac{C \cdot 50}{8} \log(n) \ln(n) n^{C\alpha}$ contains a homologous anchor. This lemma establishes that homologous anchors are ``dense''. Denote the space under which there are no homologous gaps of size $\ge g(n)$ by $F2$: we will be in the space $F2$ with probability $\ge 1 - \frac{1}{n}$. 

\begin{lemma}[Defining the $\bm{F2}$ space]\label{hom_breaks:f2}
With probability $\geq 1 - \frac{1}{n}$, no homologous gap in $S[p+1:p+m']$ has size greater than
\vspace{-1em}
\[
g(n) = \frac{50k}{8(1 - \theta_T)^k} \ln(n) = \frac{C \cdot 50}{8} \log(n) \ln(n) \cdot n^{C\alpha}
\]
plus a small $C \log n$ term we will ignore because it is small asymptotically.
\end{lemma}
\begin{proof}
    The proof is identical to Lemma 6 from Yu and Shaw \cite{shaw2023proving}. We will replicate the proof for this crucial lemma. 

    Let $\ell = g(n) = \frac{50k\ln(n)}{8(1-\theta_T)^k}$. Define $\mathrm{HG_1}, \dots, \mathrm{HG}_{m' - \ell + 1}$ be the random variables indicating if there is a homologous gap of length $\ell$ at a given position, i.e., $\mathrm{HG}_i = 1$ when no $k$-mer in $S[p+j:p+j+\ell-1]$ is part of a homologous anchor. Note that $\mathbb{E}[\sum_{i=1}^{m'-\ell+1}HG_i] \le \frac{m'}{n^2} \le \frac{1}{n}$. Applying Markov's inequality, we get that $\mathrm{Pr}(\sum_{i=1}^{m'-\ell+1}\mathrm{HG}_i \ge 1) = \mathrm{Pr}(\sum_{i=1}^{m'-\ell+1}\mathrm{HG}_i \ge n \cdot 1/n) \le \frac{1}{n}$, and the result follows.
\end{proof}

The previous lemma shows that within the high probability $\bm{F2}$ space, there are no homologous gaps in $S[p+1:p+m']$ of size greater than $O(g(n))$; homologous anchors are ``dense'' under $F2$. In the following lemma, we work under the intersection $\bm{G} = \bm{F1} \land \bm{EC} \land \bm{F2}$ to show that there are at most $O(g(n))$ unrecovered points before the first anchor or after the last anchor of any optimal chain.

\begin{lemma}[Working in $\bm{G}$]\label{lemma:missed_ends}
    For any optimal chain $\mathcal{C} = ((i_1,j_1), \dots, (i_u,j_u))$, there are at most $O(g(n))$ unrecovered points on $P_H$ before the first anchor of $\mathcal{C}$ or after the last anchor. Formally, let $S = \{(x,y) \in P_H \mid (x \le i_1 \lor y \le j_1) \lor (x \ge i_u \lor y \ge j_u)\}$. Then $|S| = O(g(n))$.
\end{lemma}
\begin{proof}
    Consider the first anchor $(i_1, j_1)$. Working in $\bm{G}$, there are no spurious anchors, so $(i_1, j_1)$ must be a homologous anchor or a clipping anchor. In either case, there is a point on the anchor that belongs to the homologous path. Call this point $(i, f(i))$. Under $F2$, there are no homologous gaps of size $\ge g(n)$, so there is some homologous anchor occurring at $i_0 < i_1$, $A(i,f(i)) = 1$, such that $i - i_0 \le g(n)$. Under $EC$, since $i - i_0 \le g(n)$, there cannot be more than $\frac{1}{t_0}(\frac{2}{\beta} + 1)g(n)$ insertions between $f(i_0)$ and $f(i)$. Thus, $f(i) - f(i_0) \le 2\frac{1}{t_0}(\frac{2}{\beta} + 1)g(n)$. It follows that $(i - i_0) + (f(i) - f(i_0)) \le 3\frac{1}{t_0}(\frac{2}{\beta} + 1)g(n)$. Adding this homologous anchor to the chain changes its score by at least $1 - \xi(i - i_0 + f(i) - f(i_0))$, which is greater than $0$, since $\xi = \frac{1}{n}$ and $n \gg g(n)$. Thus, adding it to the chain increases its score. It follows that the first anchor $(i_1, j_1)$ in any optimal chain has some point on it at most $g(n)$ away from the start of the generative position, $p+1$. We once again have that the $y$-value of this point can be at most $2\frac{1}{t_0}(\frac{2}{\beta} + 1)g(n)$  larger than the first point on the homologous path. Thus, at most $g(n) + 2\frac{1}{t_0}(\frac{2}{\beta} + 1)g(n) \le 3\frac{1}{t_0}(\frac{2}{\beta} + 1)g(n)$ points are missed before this first anchor. The same logic shows that there can be at most $3\frac{1}{t_0}(\frac{2}{\beta} + 1)g(n)$ points missed on the homologous path after the last anchor. Combining these facts, at most $6\frac{1}{t_0}(\frac{2}{\beta} + 1)g(n) = O(g(n))$ points are missed before or after any optimal chain.
    
\end{proof}

\section{Recoverability and Runtime Theorems}\label{section:rec-theorem}

Recall that in the previous sections, we defined three high probability spaces \textbf{(Informal)}:
\begin{itemize}[nosep, topsep=0pt, leftmargin=*]
    \item $\bm{[F1]}$, defined in Lemma \ref{lemma:f1}: For large enough $n$, there are no spurious anchors at all with probability $\ge 1 - \frac{3}{n}$.
    \item $\bm{[EC]}$, defined in Lemma \ref{lemma:EC_space}: Each $\Theta(k)$-block in $S$ corresponds to an $\Theta(k)$-block in $S'$ after edits with prob.~$\ge 1 - \frac{2}{n}$.
    \item $\bm{[F2]}$, defined in Lemma \ref{hom_breaks:f2}: With probability $\geq 1 - \frac{1}{n}$, no homologous gap in $S[p+1:p+m']$ has size greater than $O(g(n)) = O(\frac{C \cdot 50}{8} \log(n) \ln(n) \cdot n^{C\alpha})$.
\end{itemize}

As before, we will analyze seed-chain-extend under the intersection of these spaces, $G = F1 \land F2 \land EC$. Recall again that $G$ is itself a high probability space: $\Pr(G^c) \le \Pr(F1^c) + \Pr(F2^c) + \Pr(EC^c) \le \frac{6}{n} \implies \Pr(G) \ge 1 - \frac{6}{n}$.


We now move on to the first main result: proving the expected recoverability of seed-chain-extend with indels is $\ge 1 - O\Bigl(\frac{1}{\sqrt{m}}\Bigr)$ for large enough $n$. To this end, we lower bound the recoverability of any chain while working in $G$.

\begin{lemma}[Working in $\bm{G}$]\label{lemma: recov_lower_bound}
    Let $\mathcal{C} = ((i_1, j_1), \dots, (i_u, j_u))$ be an optimal chain. Let the first and last points of the homologous path be $(i_s, j_s), (i_e, j_e)$, respectively. The recoverability of $\mathcal{C}$ can be lower bounded as

    $$R(\mathcal{C}) \ge 1 - \frac{(i_1 - i_s) + (j_1 - j_s) + (i_e - i_u) + (j_e - j_u)}{|P_H|}.$$

\end{lemma}

\begin{proof}
    Any point $(x,y) \in P_H$ for which $0 \le x < i_1$ or $0 \le y < j_1$ or $x \ge i_u + k$ or $y \ge j_u + k$ is clearly not recoverable. There are at most $i_1 + j_1 + (i_e - i_u) + (j_e - j_u)$ such points. Consider any point $(x,y) \in P_H$ such that $i_1 \le x \le i_u$ and $j_1 \le y \le j_u$. These points belong to three categories:
    \vspace{-.9em}
    \begin{enumerate}
        \item \textbf{Anchor Points}. Since $(x,y)$ lies on an anchor, it is recovered.
        \item \textbf{Between anchors}. There exist anchors $(i_p,j_p)$ and $(i_{p+1},j_{p+1})$ such that $i_p + k - 1 \le x \le i_{p+1}$ and $j_p + k - 1 \le y \le j_{p+1}$. Then $(x,y)$ lies in an extension box and is recovered.
        \item \textbf{Overlap with clipping anchors}. $(x,y)$ belongs to a region covered by a clipping anchor. Specifically, there exists a clipping anchor $(i_\ell, j_\ell)$ such that $i_\ell \le x \le i_\ell + k - 1 \lor j_\ell \le y \le j_\ell + k - 1$. Since $(x,y)$ is not recovered, it does not lie on the clipping anchor itself. Thus, it belongs to a non-recoverable region, which is covered by $U$.
    \end{enumerate}
    \vspace{-1em}
    Thus, each point contained within the bounds of the chain is either recovered or part of an unrecoverable region, in which case it is contained in $U$, and the result follows.
\end{proof}

Finally, we give the main recoverability result: 

\begin{theorem}[Recoverability theorem]
    The expected recoverability of an optimal chain, $\mathcal{C} = ((i_1, j_1), \dots, (i_u, j_u))$, is $\ge 1 - O\Bigl(\frac{1}{\sqrt{m}}\Bigr)$ for large enough $n$.
\end{theorem}

\begin{proof}
    Recall that the recoverability of a chain $\mathcal{C}$ is defined as $R(\mathcal{C}) = \frac{|Align(\mathcal{C}) \cap P_H\setminus U|}{|P_H \setminus U|}$. We showed this is $ \ge 1 - \frac{(i_1 - i_s) + (j_1 - j_s) + (i_e - i_u) + (j_e - j_u)}{|P_H|}$, in $G = F1 \land F2 \land EC$  in the previous lemma. Note that $\Pr(G) \ge 1 - \frac{2}{n} - \frac{3}{n} - \frac{1}{n} = 1 - \frac{6}{n}$. 

    Working in $G$: by Lemma \ref{lemma:missed_ends}, the number of points missed before the start of the chain or after the end of the chain is $\le (i_1 - i_s) + (j_1 - j_s) + (i_e - i_u) + (j_e - j_u) \le  6\frac{1}{t_0}(\frac{2}{\beta} + 1)g(n)$. Recall that $g(n) = \frac{C \cdot 50}{8} \log(n) \ln(n) n^{C\alpha}$. For large enough $n$, $n^{C\alpha} < \sqrt{m}$, so $\frac{g(n)}{m} \le O(\frac{1}{\sqrt{m}})$. Thus,  $\mathbb{E}(\frac{(i_1 - i_s) + (j_1 - j_s) + (i_e - i_u) + (j_e - j_u)}{|P_H|} \mid G) = O(\frac{g(n)}{m}) = O(\frac{1}{\sqrt{m}})$.
    
    Combining it all,
\vspace{-1em}\[
\mathbb{E}(R) \ge \mathbb{E}(R \mid G) \Pr(G) 
= \Bigl( 1 - O\Bigl(\frac{1}{\sqrt{m}}\Bigr ) \Bigr) (1 - 6/n)
= 1 - O\Bigl(\frac{1}{\sqrt{m}}
\Bigr).
\]
\vspace{-0.5em}
\end{proof}
\vspace{-1em}
%
%
\vspace{-0.5em}
Finally, we now prove that the runtime of seed-chain-extend for any optimal chain is $O(mn^{C\alpha}\log n)$ in expectation. Here, we assume that the reference string $S$ is already seeded. This follows in two parts: (1) we first establish that the chaining runtime is $O(mn^{C\alpha}\log n)$ by showing that $\mathbb{E}[N\log N] = O(mn^{C\alpha}\log n)$, where $N$ is the number of anchors between $S,S'$; (2) we show that the extension runtime between anchors in an optimal chain is upper bounded by $O(mn^{C\alpha}\log n)$. Since the runtime of seed-chain-extend is dominated by chaining with extension, this proves the result. To show this second point, we prove that extension runtime through any optimal chain is at most the extension runtime through a related chain containing only homologous anchors with no mutations. See Appendix \ref{sec:runtime-thm-proof} for details.

\begin{theorem}[Runtime theorem]\label{thm:runtime_theorem}
    Given that the reference string $S$ is already seeded, the expected runtime $\mathbb{E}[T_{SCE}]$ of seed-chain-extend is $O(mn^{C\alpha}\log n)$.
\end{theorem}
\begin{proof}
        The runtime $T_{SCE}$ of any chain $\mathcal{C}$ is
        the runtime of seeding the query $S'$, chaining, and extension. Seeding $S'$ is fast: it takes $O(m)$ time.

        We have that $T_{SCE} = O(m) + T_{Chain} + T_{Ext}$. By Lemma \ref{lemma:chaining_runtime}, $\mathbb{E}[T_{Chain}] = O(mn^{C\alpha}\log n)$, and by Lemma \ref{lemma:extension_runtime}, $\mathbb{E}[T_{Ext}] = O(mn^{C\alpha}\log n)$. Thus, $\mathbb{E}[T_{SCE}] = O(m) + O(\mathbb{E}[T_{Chain}]) + O(\mathbb{E}[T_{Ext}]) = O(mn^{C\alpha}\log n)$.
\end{proof}


\section{Experimental results}



\begin{figure}[t]
  \centering
  \begin{subfigure}{0.48\textwidth}
    \centering
    \includegraphics[height=0.3\textheight,keepaspectratio]
      {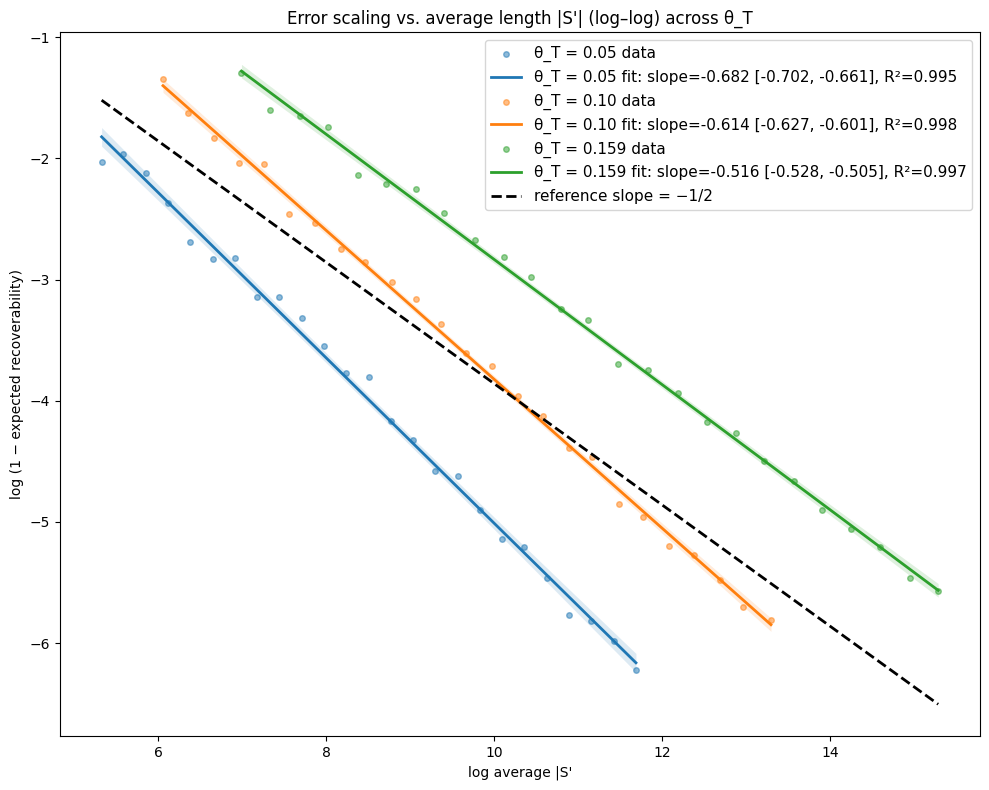}
    \caption{Log–log fit vs. avg. length of $S'$ for $\theta_T \in \{0.05,0.10,0.159\}$.}
    \label{fig:recoverability_by_theta}
  \end{subfigure}
  \hfill
  \begin{subfigure}{0.48\textwidth}
    \centering
    \includegraphics[height=0.3\textheight,keepaspectratio]
      {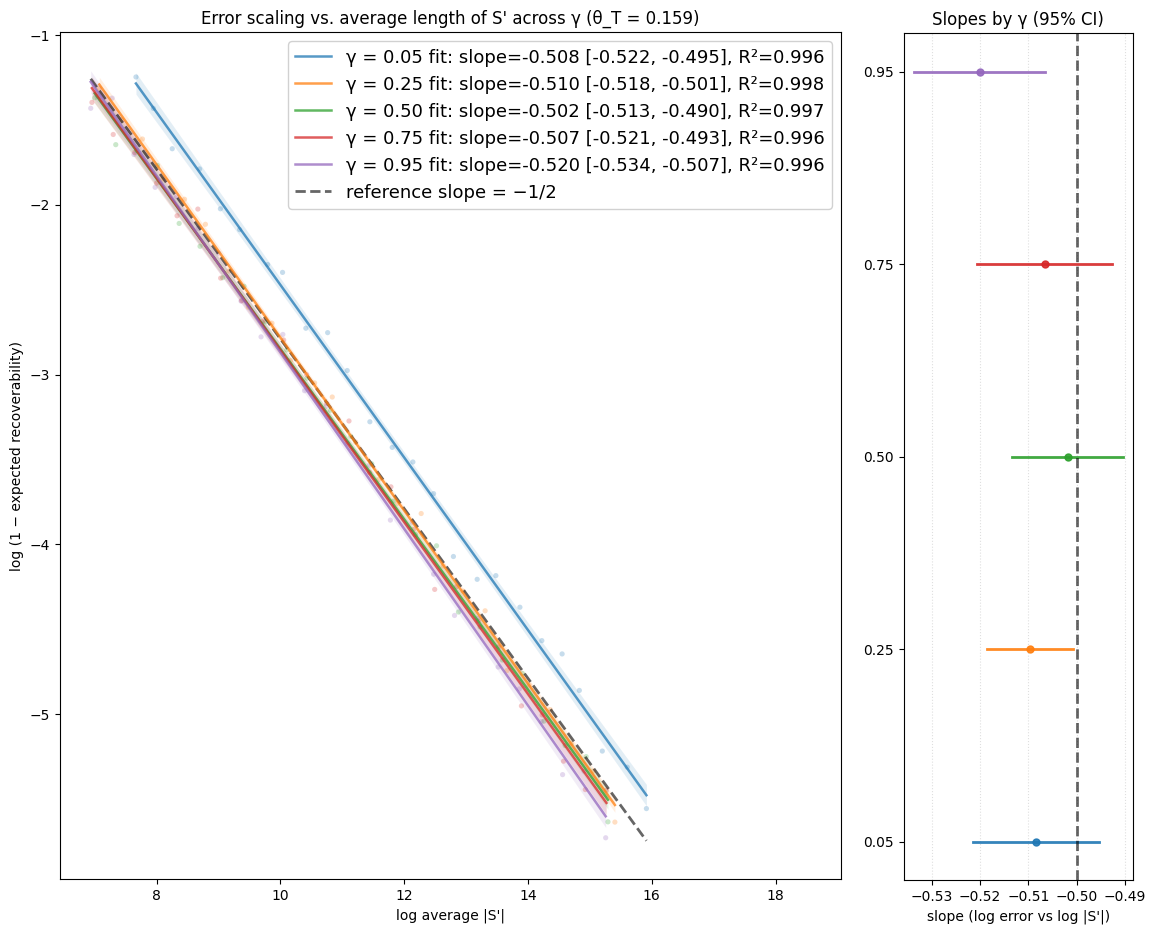}
    \caption{Left: log–log fit across $\gamma \in \{0.05,0.25,0.50,0.75,0.95\}$. Right: slopes (95\% CIs) with reference line at $-0.50$.}
    \label{fig:recoverability_by_gamma}
  \end{subfigure}

  \begin{subfigure}{1\textwidth}
    \centering
    \includegraphics[width=0.5\linewidth]{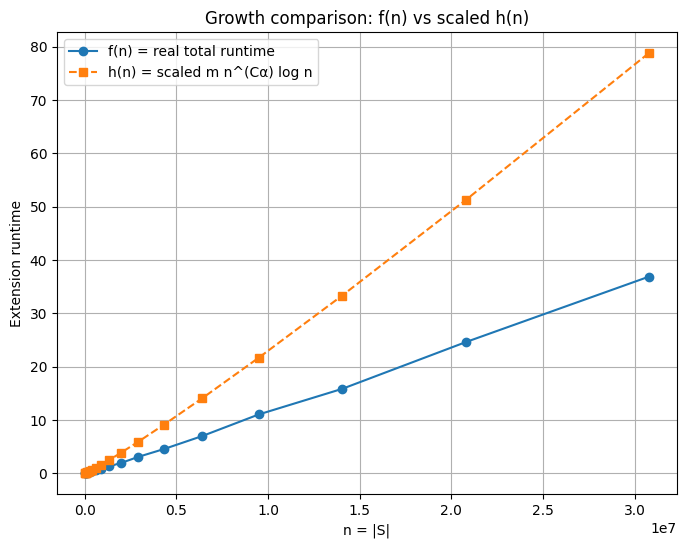}
    \caption{The average empirical runtime (blue) and the scaled theoretical prediction $O(mn^{C\alpha}\log n)$ (orange) for $26 \le k \le 44$. The growth of the empirical average runtime is consistent with $\mathbb{E}(T_{SCE}) = O(mn^{C\alpha}\log n)$.}
    \label{fig:runtime_experiments}
  \end{subfigure}
  \caption{Recoverability analyses across mutation parameters and corresponding runtime scaling behavior.}
  \label{fig:recoverability_and_runtime}
\end{figure}

Following the mutation model defined in Section \ref{sec:preliminaries}, we perform experiments in the following way: over $100$ iterations for each $k$ ranging from $20$ to $44$, we set $n$ such that $k = C\log_4 (n)$. We define $C = \frac{3}{1-2\alpha}$ and choose $m = n^{\frac{2C\alpha+1}{2}}$, so that $m = \Omega(n^{2C\alpha + \epsilon}) \ll n$. For each experiment, we produce an $S$ by randomly sampling letters i.i.d. from $\{A,C,T,G\}$ and, passing a substring of $S$ through the indel channel defined in the mutation model with some value of $\theta_T$ and $\gamma$ to get both the homologous path $P_H$ and the mutated substring $S'$. Since $\theta_T$ is an upper bound on the total mutation rate, we randomly select $\theta_i, \theta_d, \theta_s$ such that $\theta_i + \theta_d + \theta_s = \theta_T$.

Our experiments center around estimating the rate of convergence of the expected recoverability $\mathbb{E}(R)$ in terms of $m = |S'|$. We do this is by collecting the average recoverability over $100$ trials for fixed values of $\theta_T, k$, and $\gamma$. We also take the average length of $S'$ over the $100$ trials for fixed values of $k, \theta_T$, and $\gamma$. We then fit a linear regression on $\log(1 - \hat{\mathbb{E}}(R))$ against $\log{{\mathbb{\hat{E}}(|S'|)}}$. For example, if $\log(1 - \hat{\mathbb{E}}(R)) \approx a + b\log{{\mathbb{\hat{E}}(|S'|)}}$, then $1 - \hat{\mathbb{E}}(R) \approx c(\mathbb{\hat{E}}(|S'|))^b$, from which we conclude that $\hat{\mathbb{E}}(R) \approx 1 - O(\mathbb{\hat{E}}(m))^b)$. It is necessary to take the average value of $|S'|$ over iterations since the length of $S'$ is a random variable controlled through the sequence of edits.
The number of iterations required to see repeated instances of the same length of $S'$ is infeasible.
Instead,
Under $EC$, $|S'|$ and $|S|$ are interchangeable up to a constant.
In the log-log plot, the constant factor is absorbed into the intercept, leaving the slope unchanged.

We present experimental results in this section that show that the error $1 - \mathbb{E}(R)$ goes to $0$ at least as quickly as $\frac{1}{\sqrt{m}} = |S'|^{-\frac{1}{2}}$. Note that this implies that $\mathbb{E}(R) = 1 - O(\frac{1}{\sqrt{m}})$. 
Fig.~\ref{fig:recoverability_by_theta} shows the rate of convergence of $1 - \mathbb{E}(R)$ to $0$ as a function of $\theta_T$. It is clear from the experimental results that convergence is at least as quick as $\frac{1}{\sqrt{m}}$ for each value of $\theta_T \in \{0.05, 0.10, 0.159\}$. This represents a range of values for $\theta_T$, which establishes that, indeed, for all $\theta_T \le 0.159$, we have that $\mathbb{E}(R) = 1 - O(\frac{1}{\sqrt{m}})$. 
Fig.~\ref{fig:recoverability_by_gamma} shows $\mathbb{E}(R) = 1 - O(\frac{1}{\sqrt{m}})$ is independent of the value of $\gamma$, provided that $\gamma < 1$. We show log-log plots of recoverability error against the average $m = |S'|$ for $\gamma \in \{0.05, 0.25, 0.50, 0.75, 0.95\}$. The slopes for each line are $\le -0.50$, establishing this result over a range of values for $\gamma$.


The experimental runtime results are presented in Fig.~\ref{fig:runtime_experiments}, where we fix $\gamma = 0.5, \theta_T = 0.10, m' = n^{C\alpha + \frac{1}{2}}$, and vary $k$ from $26$ to $44$. Similar to the recoverability experiments, $S$ and $S'$ are generated $100$ times for each choice of parameters, and the average total runtime (chaining $+$ extension) and average predicted runtime are recorded. The predicted runtime is scaled by $\frac{2f(n_0)}{g(n_0)} = 3.87 \times 10^{-8}$ to clearly show the growths. Fig.~\ref{fig:runtime_experiments} empirically supports that the expected total runtime of seed chain extend is $O(mn^{C\alpha}\log n)$.



\section{Conclusion}

We have shown that under moderate assumptions, the expected recoverability of an optimal chain is $\ge 1 - O(\frac{1}{\sqrt{m}})$ and the expected runtime is $O(mn^{C\alpha}\log n) \le O(mn^{3.15 \cdot \theta_T}\log n)$, for a mutation model allowing indels and substitutions. This extends the paper of Yu and Shaw \cite{shaw2023proving},  which addressed the substitution-only regime. The main technical contributions of this paper are the introduction of \textit{clipping} anchors, which lie partially on the homologous path, the application of concentration inequalities to bound the probability of unfavorable edit histories, and a more complete understanding of why the seed-chain-extend heuristic is successful.


%
%
%
\bibliography{main.bib}

\begin{thebibliography}{10}

\bibitem{britten2002divergence}
Roy~J Britten.
\newblock Divergence between samples of chimpanzee and human dna sequences is 5\%, counting indels.
\newblock {\em Proceedings of the National Academy of Sciences}, 99(21):13633--13635, 2002.

\bibitem{koonin2000impact}
Eugene~V Koonin, L~Aravind, and Alexey~S Kondrashov.
\newblock The impact of comparative genomics on our understanding of evolution.
\newblock {\em Cell}, 101(6):573--576, 2000.

\bibitem{nurk2022complete}
Sergey Nurk, Sergey Koren, Arang Rhie, Mikko Rautiainen, Andrey~V Bzikadze, Alla Mikheenko, Mitchell~R Vollger, Nicolas Altemose, Lev Uralsky, Ariel Gershman, et~al.
\newblock The complete sequence of a human genome.
\newblock {\em Science}, 376(6588):44--53, 2022.

\bibitem{siren2021pangenomics}
Jouni Sir{\'e}n, Jean Monlong, Xian Chang, Adam~M Novak, Jordan~M Eizenga, Charles Markello, Jonas~A Sibbesen, Glenn Hickey, Pi-Chuan Chang, Andrew Carroll, et~al.
\newblock Pangenomics enables genotyping of known structural variants in 5202 diverse genomes.
\newblock {\em Science}, 374(6574):abg8871, 2021.

\bibitem{berger2020levenshtein}
Bonnie Berger, Michael~S Waterman, and Yun~William Yu.
\newblock Levenshtein distance, sequence comparison and biological database search.
\newblock {\em IEEE transactions on information theory}, 67(6):3287--3294, 2020.

\bibitem{Chvátal_Sankoff_1975}
Vacláv Chvátal and David Sankoff.
\newblock Longest common subsequences of two random sequences.
\newblock {\em Journal of Applied Probability}, 12(2):306–315, 1975.

\bibitem{kiwi2005expected}
Marcos Kiwi, Martin Loebl, and Ji{\v{r}}{\'\i} Matou{\v{s}}ek.
\newblock Expected length of the longest common subsequence for large alphabets.
\newblock {\em Advances in Mathematics}, 197(2):480--498, 2005.

\bibitem{reinert2000probabilistic}
Gesine Reinert, Sophie Schbath, and Michael~S Waterman.
\newblock Probabilistic and statistical properties of words: an overview.
\newblock {\em Journal of Computational Biology}, 7(1-2):1--46, 2000.

\bibitem{ukkonen1983approximate}
Esko Ukkonen.
\newblock On approximate string matching.
\newblock In {\em International Conference on Fundamentals of Computation Theory}, pages 487--495. Springer, 1983.

\bibitem{needleman1970general}
Saul~B Needleman and Christian~D Wunsch.
\newblock A general method applicable to the search for similarities in the amino acid sequence of two proteins.
\newblock {\em Journal of molecular biology}, 48(3):443--453, 1970.

\bibitem{smith1981identification}
Temple~F Smith, Michael~S Waterman, et~al.
\newblock Identification of common molecular subsequences.
\newblock {\em Journal of molecular biology}, 147(1):195--197, 1981.

\bibitem{four_russians}
A.~V. Aho, M.~R. Garey, and J.~D. Ullman.
\newblock The transitive reduction of a directed graph.
\newblock {\em SIAM Journal on Computing}, 1(2):131--137, 1972.

\bibitem{backurs2015edit}
Arturs Backurs and Piotr Indyk.
\newblock Edit distance cannot be computed in strongly subquadratic time (unless seth is false).
\newblock In {\em Proceedings of the forty-seventh annual ACM symposium on Theory of computing}, pages 51--58, 2015.

\bibitem{altschul1990basic}
Stephen~F Altschul, Warren Gish, Webb Miller, Eugene~W Myers, and David~J Lipman.
\newblock Basic local alignment search tool.
\newblock {\em Journal of molecular biology}, 215(3):403--410, 1990.

\bibitem{van2014top}
Richard Van~Noorden, Brendan Maher, and Regina Nuzzo.
\newblock The top 100 papers.
\newblock {\em Nature News}, 514(7524):550, 2014.

\bibitem{ukkonen1985algorithms}
Esko Ukkonen.
\newblock Algorithms for approximate string matching.
\newblock {\em Information and control}, 64(1-3):100--118, 1985.

\bibitem{myers1986nd}
Eugene~W Myers.
\newblock An o (nd) difference algorithm and its variations.
\newblock {\em Algorithmica}, 1(1):251--266, 1986.

\bibitem{chaisson2012mapping}
Mark~J Chaisson and Glenn Tesler.
\newblock Mapping single molecule sequencing reads using basic local alignment with successive refinement (blasr): application and theory.
\newblock {\em BMC bioinformatics}, 13:1--18, 2012.

\bibitem{groot2022exact}
Ragnar Groot~Koerkamp and Pesho Ivanov.
\newblock Exact global alignment using a* with chaining seed heuristic and match pruning.
\newblock {\em bioRxiv}, pages 2022--09, 2022.

\bibitem{ivanov2022fast}
Pesho Ivanov, Benjamin Bichsel, and Martin Vechev.
\newblock Fast and optimal sequence-to-graph alignment guided by seeds.
\newblock In {\em International Conference on Research in Computational Molecular Biology}, pages 306--325. Springer, 2022.

\bibitem{langmead2012fast}
Ben Langmead and Steven~L Salzberg.
\newblock Fast gapped-read alignment with bowtie 2.
\newblock {\em Nature methods}, 9(4):357--359, 2012.

\bibitem{li2018minimap2}
Heng Li.
\newblock Minimap2: pairwise alignment for nucleotide sequences.
\newblock {\em Bioinformatics}, 34(18):3094--3100, 2018.

\bibitem{bukh2020length}
Boris Bukh and Raymond Hogenson.
\newblock Length of the longest common subsequence between overlapping words.
\newblock {\em SIAM Journal on Discrete Mathematics}, 34(1):721--729, 2020.

\bibitem{lember2009lcs}
J{\"u}ri Lember and Heinrich Matzinger.
\newblock {Standard deviation of the longest common subsequence}.
\newblock {\em The Annals of Probability}, 37(3):1192 -- 1235, 2009.

\bibitem{szpankowski2011average}
Wojciech Szpankowski.
\newblock {\em Average case analysis of algorithms on sequences}.
\newblock John Wiley \& Sons, 2011.

\bibitem{yu2022hyperminhash}
Yun~William Yu and Griffin~M Weber.
\newblock Hyperminhash: Minhash in loglog space.
\newblock {\em IEEE Transactions on Knowledge and Data Engineering}, 34(1):328--339, 2022.

\bibitem{ganesh2020near}
Arun Ganesh and Aaron Sy.
\newblock Near-linear time edit distance for indel channels.
\newblock In {\em 20th International Workshop on Algorithms in Bioinformatics}, page~17, 2020.

\bibitem{shaw2023proving}
Jim Shaw and Yun~William Yu.
\newblock Proving sequence aligners can guarantee accuracy in almost o (m log n) time through an average-case analysis of the seed-chain-extend heuristic.
\newblock {\em Genome Research}, 33(7):1175--1187, 2023.

\bibitem{edgar2021syncmers}
Robert Edgar.
\newblock Syncmers are more sensitive than minimizers for selecting conserved k-mers in biological sequences.
\newblock {\em PeerJ}, 9:e10805, 2021.

\bibitem{ondov2016mash}
Brian~D Ondov, Todd~J Treangen, P{\'a}ll Melsted, Adam~B Mallonee, Nicholas~H Bergman, Sergey Koren, and Adam~M Phillippy.
\newblock Mash: fast genome and metagenome distance estimation using minhash.
\newblock {\em Genome biology}, 17:1--14, 2016.

\bibitem{shaw2022theory}
Jim Shaw and Yun~William Yu.
\newblock Theory of local k-mer selection with applications to long-read alignment.
\newblock {\em Bioinformatics}, 38(20):4659--4669, 2022.

\bibitem{yu2015quality}
Y~William Yu, Deniz Yorukoglu, Jian Peng, and Bonnie Berger.
\newblock Quality score compression improves genotyping accuracy.
\newblock {\em Nature biotechnology}, 33(3):240--243, 2015.

\bibitem{janson2004large}
Svante Janson.
\newblock Large deviations for sums of partly dependent random variables.
\newblock {\em Random Structures \& Algorithms}, 24(3):234--248, 2004.

\end{thebibliography}

\newpage

\appendix

\title{Appendix to Incorporating indel channels into average-case analysis of seed-chain-extend}
\author{Spencer Gibson \and
Yun William Yu}
\titlerunning{Appendix to Average-case analysis of seed-chain-extend with indels}
\authorrunning{S. Gibson and Y.W. Yu}
\institute{Carnegie Mellon University, Pittsburgh, PA, USA}
\maketitle

\section{Dependence between $S$ and $S'$}
\label{sec:uncovered-gap-bounds}

\subsection{Tools}

In this section, we present two tools that simplify our analysis: match variables, which are the random variables for the event that a particular letter on $S$ matches some letter on $S'$, and a function that tracks where non-deleted positions on $S$ map to on $S'$. Since our analysis deals with various types of anchors, we formally define each of them in this section. 

To begin, we define the concept of a match variable. As stated before, these variables are indicator random variables for the event that positions on $S$ and $S'$ share the same letter. We can then define the random variable for the event that an anchor occurs at a position tuple of $S$ and $S'$. 

\begin{definition}
Let $M(i,j) = \mathbf{1}\{S[i] = S'[j]\}$ be the indicator random variable detecting if $S$ and $S'$ share the same character at positions $i$ and $j$. Define $A(i,j) = \prod_{\ell=0}^{k-1} M(i+\ell,\;j+\ell)$, the indicator random variable for an anchor occurring at $(i,j)$.
\end{definition}

Let $x \in [p+1, p+m']$ be an index in the generative region of $S$. If the character at position $x$ is not deleted in the mutation process, then there is a unique corresponding position in $S'$. One useful fact that we will repeatedly use is that the letter distribution at any other position in $S'$ is independent of the letter distribution at position $x$. To formalize this, we introduce the function $f$ that maps each position $i \in S$ to its unique corresponding position in $S'$ if one exists. This is done by referencing the homologous path: if $x$ is not deleted, then its corresponding position turns out to be the point $(x,y)$ on the path with the largest $y$-value. The definition below formalizes this.

\begin{definition}
Define the function $f:\{1,\dots,|S|\}\;\longrightarrow\;\{1,\dots,|S'|\}\cup\{\emptyset \}$ such that
$$
f(x)=
\begin{cases}
\emptyset, 
&\text{if }x\notin [p+1,p+m']\text{ or }x\text{ is deleted},\\[6pt]
\displaystyle
\max\bigl\{\,y : (x,y)\in P_H\bigr\},
&\text{otherwise.}
\end{cases}
$$

If in the mutation process, position $x \in [p+1,p+,m']$ of $S$ is not deleted, then either no mutation occurred at $x$ or another character is substituted for $S[x]$. In either case, there exists a position $y$ in $S'$ such that $(x,y) \in P_H$, meaning that $\max\bigl\{\,y : (x,y)\in P_H\bigr\}$ is well-defined. 

\end{definition}

Below, we provide a simple lemma stating that if $f$ takes the same value on two inputs, then it must map those inputs to $\emptyset$. In other words, $f$ is injective on preserved positions of $S$.
\begin{lemma}\label{sup_lemma:1}
    For positions $x, y \in [|S|]$, either $f(x) = f(y) = \emptyset$ or $f(x) \ne f(y)$.
\end{lemma}

\begin{proof}
    If $x$ and $y$ are both deleted, then $f(x) = f(y) = \emptyset$. Otherwise, if only one of $x$ or $y$ is deleted, without loss of generality, let it be $x$, then $f(x) = \emptyset \ne f(y)$. Lastly, if neither $x$ nor $y$ is deleted, then they have unique corresponding positions $x', y'$ on $S'$, implying that $f(x) = x' \ne y' = f(y)$.
\end{proof}

As an example, consider again $S = TACTTCGC$, which is transformed into $S'=TACTTTAC$ in Fig.~\ref{fig:combined-recoverability-matchgraph}C. Consider position $4$ on S: to calculate $f(4)$, we find the point $(4,y)$ on $P_H$ (refer to Fig.~\ref{fig:homologous_path}) with the largest $y$-value. This is the point $(4,5)$, so $f(4) = 5$. Since the letter T is inserted to the left of position $4$ and position $4$ is not deleted, the corresponding position on $S'$ is $5$. This confirms that $f(4)$ is the position on $S'$ that corresponds to the position $4$ on S. Now consider position $5$ on S: this letter is deleted when generating $S'$ so it has no corresponding position. So, $f(5) = \emptyset$, matching our intuition. 

With $f$ defined, we are equipped to analyze the dependency graph between $S$ and $S'$. We derive three key results: sufficient conditions for a set of match variables, defined in the next section, to be independent of each other, sufficient conditions for two anchors to be independent of each other, and we bound the probabilities of the three anchor types listed in Sec.~\ref{sec: match_graph}.

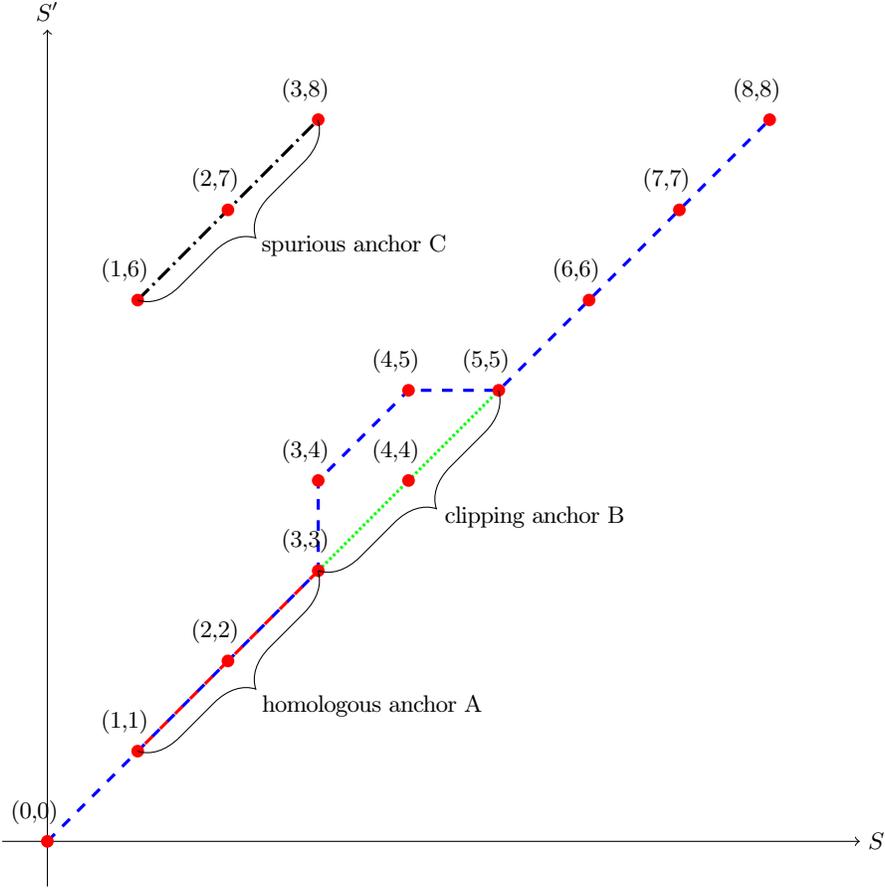
\begin{figure}[ht]
  \centering
\begin{tikzpicture}[scale=1.2]
    \draw[->] (-0.5,0) -- (9,0) node[right] {$S$};
    \draw[->] (0,-0.5) -- (0,9) node[above] {$S'$};

    \coordinate (A) at (0,0);
    \coordinate (B) at (1,1);
    \coordinate (C) at (2,2);
    \coordinate (D) at (3,3);
    \coordinate (E) at (3,4);
    \coordinate (F) at (4,4);
    \coordinate (G) at (4,5);
    \coordinate (H) at (5,5);
    \coordinate (I) at (6,6);
    \coordinate (J) at (7,7);
    \coordinate (K) at (8,8);
    \coordinate (L) at (1,6);
    \coordinate (M) at (2,7);
    \coordinate (N) at (3,8);

\draw[very thick, blue, dash pattern=on 4pt off 4pt, dash phase=0pt]
    (A) -- (B) -- (C) -- (D) -- (E) -- (G) -- (H) -- (I) -- (J) -- (K);

    \draw[very thick, red, dash pattern=on 4pt off 4pt, dash phase=4pt]
    (B) -- (D);
    \draw[very thick,green,dash pattern=on 1pt off 1pt] (D) -- (H);
    \draw[very thick,black,dash pattern=on 6pt off 2pt on 1pt off 2pt] (L) -- (N);

    \foreach \p/\label in {
    A/(0{,}0), B/(1{,}1), C/(2{,}2), D/(3{,}3), E/(3{,}4),
    F/(4{,}4), G/(4{,}5), H/(5{,}5), I/(6{,}6), J/(7{,}7),
    K/(8{,}8), L/(1{,}6), M/(2{,}7), N/(3{,}8)} {
    \fill[red] (\p) circle (2pt);
    \node[above right] at ([xshift=-0.5cm,yshift=0.1cm]\p) {\small $\label$};
}

\draw[decorate,decoration={brace,amplitude=15pt,mirror}] (B) -- (D);
\node at ($(B)!0.5!(D) + (1.6,-0.5)$) {\small homologous anchor A};

\draw[decorate,decoration={brace,amplitude=15pt,mirror}] (D) -- (H);
\node at ($(D)!0.5!(H) + (1.4,-0.4)$) {\small clipping anchor B};

\draw[decorate,decoration={brace,amplitude=15pt,mirror}] (L) -- (N);
\node at ($(L)!0.5!(N) + (1.4,-0.4)$) {\small spurious anchor C};

\end{tikzpicture}

  \caption{The points along the dashed blue line make up the homologous path given the edits turning $S = TACTTCGC$ into $S' = TACTTTAC$ following Fig.~\ref{fig:combined-recoverability-matchgraph}C. In this example, anchors are matching seeds of length $3$. Anchor A (red dash) is a \textit{homologous} anchor since it lies entirely on the path. Anchor B (green dash) is a \textit{clipping} anchor since it lies partially on the path, namely, the midpoint of the anchor does not belong to the homologous path. Anchor C (black dash) is spurious since it lies entirely off the path.}
  \label{fig:homologous_path}
\end{figure}

\subsection{Dependence structure between $S$ and $S'$}\label{sec: match_graph}

The general strategy of our proof is to show that with high probability there are no spurious anchors and that there are not large uncovered gaps in the homologous path before the first anchor or after the last anchor in any optimal chain. To do this, we will need tail bounds for the number of spurious anchors $N_S$. In turn, this requires bounding the variance of $N_S$. To show that large uncovered gaps in the homologous path do not exist, we need to show that homologous anchors are sufficiently dense. For this, we will use concentration inequalities, which require bounding the expected value of the number of homologous anchors $N_H$. To these ends, we must understand the dependence structure between $S$ and $S'$, which is the purpose of this section.


We present three main results in this section: (1) we introduce the match graph, which captures the dependence structure between $S$ and $S'$ and prove a sufficient condition for a set of match variables, which are defined shortly, to be independent, (2) we use this independence result to bound the probabilities of each anchor type occurring, and (3) we prove a sufficient condition for two anchors to be independent.

The match graph, defined below, is a bipartite graph that captures the dependence structure between positions on $S$ and $S'$. The original (also called the unconditioned) match graph contains edges between corresponding positions $i$ in $S$ and $f(i)$ in $S'$ provided that $f(i) \ne \emptyset$. The intuition is that positions connected by an edge have dependent letter distributions; this is clear in the unconditioned match graph, since $\Pr(S'[f(i)] = S[i]) = 1 - \theta_s \ne \frac{1}{\sigma}$. We refer to \textit{match variables}, where a match variable $M(i,j)$ is the event that $S[i] = S'[j]$. Fig.~\ref{fig:match_graph_cycle} shows the case where $S'$ is obtained through substitutions from $S$ and the graph is conditioned on $A(1,3), A(3,1), A(2,4), A(4,2)$: note that $S[1] \to S'[3] \to S[3] \to S'[1] \to S[1]$ is a cycle in this induced match graph. 

\begin{definition}\label{def:match-graph}
Let $\mathcal{M} = \{M(a_1,b_1), \dots, M(a_p, b_p)\}$ be a set of matching variables where $a_i$ are positions in $S$ and $b_i$ are positions in $S'$. The match graph induced by $\mathcal{M}$ refers to the graph $G = (V, E)$ where the vertices $V$ are the letters in $S$ and $S'$, and the edges are given by 
$E = \{(x_i, y_{f(x_i)}) \mid i \in (p+1, p+m') \land f(x_i) \ne \emptyset\} \cup \{(x_h, y_l) \mid M(h,l) \in \mathcal{M}\}$. Note that \(G(\mathcal{M})\) is bipartite.
\end{definition}

Intuitively, if the match variables form a cycle, then they are dependent: in Fig.~\ref{fig:match_graph_cycle}, if $A(1,1) = A(1,3) = A(3,3) = 1$, then we must also have $A(3,1) = 1$. The first main result of this section is a converse of this statement: if the match variables do not induce a cycle in the match graph then they are independent.
We will use this fact to calculate the probability that different anchor types occur, and when a pair of anchor indicator variables are independent.

We begin with a lemma from Yu and Shaw \cite{shaw2023proving}, where they prove that if random variables belong to different connected components of the match graph then they are conditionally independent.

\begin{lemma}[(Yu and Shaw) Supplemental Lemma: Conditional independence in the match graph]\label{lemma: connected_components}
Consider a set of random variables of the form $
\mathcal{M} = \{\,M(i_1,j_1),\ldots\, M(i_\ell, j_\ell)\,\}.
$
If two vertices \(u,v\) lie in separate connected components of \(G(\mathcal{M})\), then they are conditionally independent of \(\mathcal{M}\).
\end{lemma}
\begin{proof}
    Refer to Supplemental Lemma S2 of the prequel \cite{shaw2023proving} for the proof of this result.
\end{proof}

Below, we prove two lemmas that will be combined to show that match variables are independent whenever the induced match graph does not contain a cycle. The first lemma, appearing below, establishes that a set of non-corresponding variables $M_C$, those $M(i,j)$ for which $f(i) \ne j$, are conditionally independent of any set of corresponding variables given that $M_C$ does not induce a cycle in the conditional match graph. We will use this as one of the cases in 

\begin{lemma}[Adapted from Yu and Shaw]\label{lemma:independent-if-acyclic}
    Given a set of non-corresponding match variables $\mathcal{M} = \{\,M(i_1,j_1),\,M(i_2,j_2),\,\dots\},$ such that $f(i_\ell) \ne j_\ell$ for all $M(i_\ell, j_\ell) \in \mathcal{M}$ and a set of corresponding match variables $\mathcal{M}_C = \{\,M(a_1,b_1),\,M(a_2,b_2),\,\dots\},$ such that $f(a_t) = b_t$ for all $M(a_t, b_t) \in \mathcal{M}_C$, then the match variables in $\mathcal{M}$ are independent conditioned on the match variables in $\mathcal{M}_C$ if the induced match graph has no cycles. Specifically, $\Pr(\mathcal{M} \mid \mathcal{M}_C) = \prod_{\ell}^{|\mathcal{M}|} \Pr(M(i_\ell, j_\ell) \mid \mathcal{M}_C) = \prod_{\ell}^{|\mathcal{M}|}  \Pr(M(i_\ell, j_\ell))$.
\end{lemma}

\begin{proof}
    As in the prequel, denote $M(i_\ell,j_\ell) = M_\ell$ and $\mathcal{M}^{-} = \{M_1, \dots, M_n \}$ for brevity. We proceed by induction. When $|\mathcal{M}| = 0$, the statement holds trivially. Assume the statement holds up to all $n \in \mathbb{N}$. If $|\mathcal{M}| = n + 1$, we can write $\Pr(M_1,\dots,M_{n+1} \mid \mathcal{M}_C) = \Pr(\mathcal{M}^{-} \mid \mathcal{M}_C)\Pr(M_{n+1} \mid \mathcal{M}_C, \mathcal{M}^{-})$. By the inductive hypothesis, $\Pr(M_1, \dots, M_n \mid \mathcal{M}_C) = \prod_{i=1}^n \Pr(M_i)$. It remains to show that $\Pr(M_{n+1}\mid \mathcal{M}_C, \mathcal{M}^{-}) = \Pr(M_{n+1})$.

    We have that $\Pr(M_{n+1}\mid \mathcal{M}_C, \mathcal{M}^{-}) = \sum_{\alpha \in \Sigma} \Pr(S[i_{n+1}] = \alpha, S'[j_{n+1}] = \alpha \mid \mathcal{M}_C, \mathcal{M}^{-})$.

    Since the match variables do not induce a cycle, we must hve that $i_{n+1} \in S$ and $j_{n+1} \in S'$ belong to separate connected components in the conditional match graph. Thus,  the previous sum equals $\sum_{\alpha \in \Sigma} \Pr(S[i_{n+1}] = \alpha \mid \mathcal{M}_C, \mathcal{M}^{-}) \Pr(S'[j_{n+1}] = \alpha \mid \mathcal{M}_C, \mathcal{M}^{-})$.

    We will show that $\Pr(S[i_{n+1}] = \alpha \mid \mathcal{M}_C, \mathcal{M}^{-}) = \Pr(S'[j_{n+1}] = \alpha \mid \mathcal{M}_C, \mathcal{M}^{-}) = \frac{1}{\sigma}$. As before, consider the cyclic permutation of letters $C_\sigma$ on the alphabet $\Sigma$. For DNA this could be the map $A\to G\to C \to T$. Note that applying $C_\sigma$ to all positions of $S$ and $S'$ preserves the value of all match variables. As such, $C_\sigma$ is a probability preserving transformation. Thus, if we let $A$ be the set of all string pairs $(S, S')$ such that $S[i_{n+1}] = \alpha$ then $\Pr(A \mid \mathcal{M}_C, \mathcal{M}^{-}) = \Pr(C_\sigma(A) \mid \mathcal{M}_C, \mathcal{M}^{-}) = \dots = \Pr(C_\sigma^{|\Sigma| - 1}(A) \mid \mathcal{M}_C, \mathcal{M}^{-})$. The sets $A, C_\sigma(A), \dots, C_\sigma^{|\Sigma| - 1}(A)$ partition the set of string pairs, we have that $\Pr(S[i_{n+1}] = \alpha \mid \mathcal{M}_C, \mathcal{M}^{-}) = \Pr(A \mid \mathcal{M}_C, \mathcal{M}^{-}) = \frac{1}{\sigma}$. By the same argument considering pairs where $S'[j_{n+1}] = \alpha$, we have that $\Pr(S'[j_{n+1}] = \alpha \mid \mathcal{M}_C, \mathcal{M}^{-}) = \frac{1}{\sigma}$. 

    Thus, $\Pr(M_{n+1}\mid \mathcal{M}_C, \mathcal{M}^{-}) = \frac{1}{\sigma} = \Pr(M_{n+1})$, which completes the proof.
\end{proof}

Now, we show that when the match variables $\mathcal{M}$ contain only those $M(i,j)$ for which $f(i) = j$, i.e., the position $j$ in $S'$ corresponds to the position $i$ in $S$ during the generation process, then the random variables in $\mathcal{M}$ are independent.
\begin{lemma}\label{lemma:homologous_independent}
    The random variables
\[
\mathcal{M} = \{\,M(i_1,j_1),\,M(i_2,j_2),\,\dots\},
\]
where $f(i_\ell) = j_\ell$ for all $M(i_\ell, j_\ell) \in \mathcal{M}$
are independent.
\end{lemma}
\begin{proof}
    Note that conditioning on these match variables does not add new edges to the unconditioned match graph.

    Each position $i \in S$ has at most one corresponding position in $S'$ and, similarly, each position $j \in S'$ has at most corresponding position in $S$. Thus, each $M(i_\ell, f(i_\ell))$ belongs to its own connected component. By Lemma \ref{lemma: connected_components}, the match variables are independent, i.e., $\Pr(\mathcal{M}) = \prod_{\ell=1}^{|\mathcal{M}|} \Pr(M(i_\ell, j_\ell))$.
\end{proof}

Combining the two previous lemmas, we conclude that match variables $\mathcal{M}$ are independent if they do not induce a cycle in the match graph. This is detailed below.

\begin{lemma}\label{cor:match_independence}
    The random variables
\[
\mathcal{M} = \{\,M(i_1,j_1),\,M(i_2,j_2),\,\dots\},
\]
are independent if the induced match graph has no cycles.
\end{lemma}
\begin{proof}
    We first show the result holds for simple cases:

    (1) If $f(i_\ell) \ne j_\ell$ for all $M(i_\ell, j_\ell) \in \mathcal{M}$, then this lemma reduces to showing that match variables with spurious positions that do not form a cycle are independent. This is proven in Lemma \ref{lemma:independent-if-acyclic}.

    (2) If $f(i_\ell) = j_\ell$ for all $M(i_\ell, j_\ell) \in \mathcal{M}$, then we must show that match variables of corresponding positions that do not induce a cycle are independent. This is proven in Lemma \ref{lemma:homologous_independent}.

    Assume now that $\mathcal{M}$ contains both corresponding and non-corresponding match variables. Denote the set of non-corresponding match variables in $\mathcal{M}$ by $\mathcal{A} = \{M(i_\ell, j_\ell) \in \mathcal{M} \mid f(i_\ell) \ne j_\ell \}$ and the set of corresponding match variables by $\mathcal{M}$ by $\mathcal{B} = \{M(i_\ell, j_\ell) \in \mathcal{M} \mid f(i_\ell) = j_\ell \}$. For simple indexing, let $I(\mathcal{B}) = \{\ell \mid M(i_\ell, j_\ell) \in \mathcal{B}\}$ and $I(\mathcal{A}) = \{\ell \mid M(i_\ell, j_\ell) \in \mathcal{A}\}$ be the two index sets.

    We have $\Pr(A, B) = \Pr(A \mid B) \Pr(B)$. By Lemma \ref{lemma:homologous_independent}, $\Pr(B) = \prod_{t \in I(\mathcal{B})} \Pr(M(i_t,j_t))$. By Lemma \ref{lemma:independent-if-acyclic}, $\Pr(A \mid B) = \prod_{\ell \in I(\mathcal{A})} \Pr(M(i_\ell, j_\ell))$. 
    
    Thus, $\Pr(\mathcal{M}) = \prod_{\ell \in I(\mathcal{A})} \Pr(M(i_\ell, j_\ell))\prod_{t \in I(\mathcal{B})}\Pr(M(i_t,j_t)) = \prod_{p = 1}^{|\mathcal{M}|} \Pr(M(i_p, j_p))$, concluding the proof.
\end{proof}

We now move to the second main result of this section, which is that the match graph induced by a single anchor is acyclic. From this it follows that the match variables present in an anchor are independent. Using this, we will bound the probabilities of each anchor type.

\begin{lemma}\label{cor:match_graph_single_anchor}
    The match graph induced by a single anchor, $A(i,j)$, has no cycles.
\end{lemma}
\begin{proof}
    First note that any position in $S$ that is deleted and any position in $S'$ that is inserted can have degree at most $1$, coming from the anchor $A(i,j)$. Any such position cannot be involved in a cycle, so we can remove all such positions from the graph. Relabel the surviving positions in $S$ and $S'$ as their new indices. The interval $(i,i+k-1)$ becomes $X_{i'}$, possibly empty, and $(j,j+k-1)$ becomes $X_{j'}$, also possibly empty. If $X_{i'}$ is empty then there cannot be any cycle that uses positions on $S$ since each node has degree at most $1$, and hence there are no cycles since the graph is bipartite; we conclude the same if $X_{j'}$ is empty. Note that the unconditioned match graph at this point appears as two sets of vertices of equal size with edges between each corresponding pair.
    
    We now assume $X_{i'}$ and $X_{j'}$ are nonempty and that there exists a cycle. Let $i' = \min H_{i'}$ and $j' = \min X_{j'}$. First, suppose that $i' < j'$. Let $x_{i+l}$ be the first point on $S$ belonging to a cycle with $l \ge 0$. Then $x_{i+l}$ has the neighbor $y_{i+l}$, so let the first edge of the cycle be $(x_{i+l}, y_{i+l})$. Then $y_{i+l}$ must have degree exactly $2$ for there to exist a cycle and its remaining edge must be induced from the anchor. Since $i' < j'$, its neighbor lies to the left of $x_{i+l}$, it is some $x_{i+a}$ where $0 \le a < l$, contradicting the minimality of $l$. Similarly, if $i' > j'$, we can apply the same argument but for the largest $l$ such that $x_{i+l}$ is in a cycle. Thus, no cycle exists.
\end{proof}

The above lemma allows us to conclude that the probability a spurious anchors is $\frac{1}{\sigma^k}$.

\begin{corollary}[Spurious anchor probability]
\label{cor:anchor-prob}
     Let $(i,j) \in |S| \times |S'|$: if $(i,j)$ is such that $\{(i+l,j+l) \mid 0 \le l \le k-1\} \cap P_H = \emptyset$, then $\Pr(A(i,j)) = \frac{1}{\sigma^k}$.
\end{corollary}
\begin{proof}
    Lemma \ref{cor:match_graph_single_anchor} shows that the match graph conditioned on a single anchor does not contain a cycle, which means that for an anchor $\Pr(A(i,j)) = \prod_{t=0}^{k-1} \Pr(M(i+t, j+t))$. 

    If $(i,j)$ represents the start of a spurious anchor, then $\Pr(M(i+t, j+t)) = \frac{1}{\sigma}$ for each $0 \le t \le k-1$. Combining terms proves the lemma.
\end{proof}

We now give sufficient conditions for two anchors $A(i,j), A(h,l)$ to be independent; that is, $\Pr(A(i,j), A(h,l)) = \Pr(A(i,j))\Pr(A(h,l))$.  Intuitively, the first condition $|i - h| \ge k$ or $|j - l| \ge k$ ensures that the anchors do not overlap too much -- the anchors' coverage can overlap on $S$ or $S'$ but not on both. The second condition prevents ``twisting'': consider anchors $A(1,3)$ and $A(3,1)$ in the substitution-only mutation model as shown in Fig.~\ref{fig:match_graph_cycle}. There exists a cycle going from $x_1 \to y_3 \to x_3 \to y_1 \to x_1$ where $x_i$ represents node $i$ on the top set of vertices and $y_i$ represents node $i$ on the bottom set of vertices. The biological interpretation is that under low mutation rates, $x_1$ and $y_1$ are likely to be equal, as are $x_3$ and $y_3$. This implies that the events $x_1 = y_3$ and $x_1 = y_3$ are not independent. 

\begin{figure}[ht]
\centering
\begin{tikzpicture}[
    every node/.style={
        circle,
        draw=black,
        fill=white,
        minimum size=12pt,
        line width=0.8pt,
        inner sep=0pt
    },
    edge/.style={
        line width=1.2pt
    }
]

\def\yTop{1.6}
\def\yBottom{0}

\def\xStep{2.0}

\node (t1) at (0*\xStep,\yTop) {};
\node (t2) at (1*\xStep,\yTop) {};
\node (t3) at (2*\xStep,\yTop) {};
\node (t4) at (3*\xStep,\yTop) {};

\node (b1) at (0*\xStep,\yBottom) {};
\node (b2) at (1*\xStep,\yBottom) {};
\node (b3) at (2*\xStep,\yBottom) {};
\node (b4) at (3*\xStep,\yBottom) {};

\draw[edge,blue] (t1) -- (b3);
\draw[edge,blue] (t2) -- (b4);

\draw[edge,red] (t3) -- (b1);
\draw[edge,red] (t4) -- (b2);

\draw[edge,dashed] (t1) -- (b1);
\draw[edge,dashed] (t2) -- (b2);
\draw[edge,dashed] (t3) -- (b3);
\draw[edge,dashed] (t4) -- (b4);

\end{tikzpicture}
\caption{Induced match graph in the substitution-only regime of an initial string of length $4$ with anchors $A(1,3)$ and $A(3,1)$. These anchors violate Yu and Shaw's \cite{shaw2023proving} conditions for independence and, as can be seen, there exists a cycle in the graph.}
\label{fig:match_graph_cycle}
\end{figure}
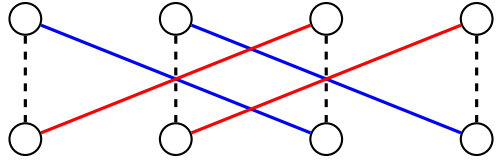



\begin{lemma}[General independence lemma]\label{indel-independence}
For \( A(i, j) \) and \( A(h, l) \), if both of the following conditions hold:
\begin{enumerate}
    \item \( |i - h| \ge k \) or \( |j - l| \ge k \), and
    \item \( [i:i+k-1] \cap f^{-1}([l:l+k-1]) = \emptyset \) or \( [h:h+k-1] \cap f^{-1}([j:j+k-1]) = \emptyset \),
\end{enumerate}
then the induced match graph on the \( M \) variables for \( A(i, j) \) and \( A(h, l) \) has no cycles.
\end{lemma}
\begin{proof}
    There are four cases to consider. We proceed with the first case, $|i - h| \ge k$ and $f[i:i+k-1] \cap [l:l+k-1] = \emptyset$. The remaining cases follow symmetrically.

    Let $X_i = [i:i+k-1]$. Since $|i - h| \ge k$, we have that $[i:i+k-1] \cap [j:j+k-1] = \emptyset$, which implies that $f([i:i+k-1]) \cap f([h:h+k-1]) = \emptyset$ by Lemma \ref{sup_lemma:1}. Furthermore, since $f([i:i+k-1]) \cap [l:l+k-1] = \emptyset$, we have that for any $x \in X_i$, $f(x) \notin f([h:h+k-1]) \cup [l:l+k-1]$. Thus, $f(x)$ can have degree at most $2$: one edge from the unconditioned match graph and one due to conditioning on $A(i,j)$. Similarly, all edges touching $x$ must come from the original unconditioned match graph or from conditioning on one of the anchors. Since $|i - h| \ge k$, $x$ has no edge due to conditioning on $A(h, \ell)$, and so $x$ has degree at most $2$. With this in place, we continue with the proof.
    
    Assume that there exists a cycle $C$. Let $C(X_i) \subset X_i$ be the set of all points in $X_i$ in the cycle. If $C(X_i) = \emptyset$, then there still exists a cycle after removing all edges induced by $A(i,j)$. This implies that there is a cycle in the match graph after conditioning on $A(h,\ell)$, contradicting our previous lemma. Thus, $C(X_i)$ is not empty.

    We first show that for any $x \in C(X_i)$, $x$ cannot be deleted. This is clear since if $x$ were deleted then it has no edge in the unconditioned match graph and so can only have degree $1$ due to some edge coming from conditioning on $A(i,j)$; since a vertex with degree $1$ cannot be in a cycle, it follows $x \notin C(X_i)$. This is also establishes that $f(x)$ is not null.

    Now, note that each $x_p \in C(X_i)$ has exactly two neighbors in the cycle: $f(x)$ and $y_{j+p-i}$, the neighbor from conditioning on $A(i,j)$. Specifically, there are $|C(X_i)|$ neighbors due to $A(i,j)$ in $C(X_i)$, each of which must be in the cycle. Each position in $f(C(X_i))$ has a neighbor in the cycle due to $A(i,j)$ and there are exactly $|f(C(X_i))| = |C(X_i)|$ such neighbors. Since each neighbor of $f(C(X_i))$ due to $A(i,j)$ is in $X_i$ and in the cycle, this is exactly $C(X_i)$. This also implies that all neighbors of $C(X_i)$ due to edges from $A(i,j)$ are exactly the points in $f(C(X_i))$. 

    Considering the subgraph given by $C(X_i) \cup f(C(X_i))$. All vertices in this graph have degree exactly $2$. Since all $2$-regular graphs contain a cycle, it follows that this subgraph contains a cycle. However, it is a subgraph of the original match graph conditioned on $A(i,j)$, implying that it too contains a cycle, which cannot be true. This completes the proof.
\end{proof}

In Sec.~\ref{sec:bounding_gaps}, we bound the number of unrecovered points before the first anchor or after the last anchor in any optimal chain. This requires showing that homologous anchors are dense in the homologous path $P_H$. The key lemma in that section makes use of $k$-independence of special homologous anchors: those resulting from $k$ positions in which no mutation occurs. The last portion of this section proves this fact; namely, that the indicator variables representing no mutation occurring in $k$ consecutive positions are independent as long as they do not overlap in S.

\begin{lemma}
    Let $E_i$ be the event that there are no mutations at position $i$ for $i \in [p+1, p+m']$. If $M$ is a collection of these $E_i$ such that no two overlap in S, then the variables in $M$ are independent. 
    
    Formally, for any collection $M = \{E_i \}_{i \in S}$ where $S \subset [p+1,p+m']$, such that $\forall i, j \in S$ if $i \ne j$ then $|i - j| \ge k$, then $\Pr(\cap_{i \in S} E_i) = \prod_{i \in S}\Pr(E_i)$.
\end{lemma}
\begin{proof}
    Since mutations occur independently at each position, $\Pr(\cap_{i \in S} E_i) = \prod_{i \in S}\Pr(E_i) = ((1-\theta_i)(1-\theta_d)(1-\theta_s))^{|S|} \ge (1 - \theta_T)^{|S|}$.
\end{proof}

The following corollary is an immediate application of the previous lemma: a k-mer in the generative region of $S$ has no mutations with probability $((1-\theta_i)(1-\theta_d)(1-\theta_s))^k$. 
\begin{corollary}[No mutation in k-block probability]\label{cor:kblock_prob}
    Let $E_{i:i+k-1}$ be the event that there are no mutations in $S[i:i+k-1]$ for $i \in [p+1, p+m']$. Then $\Pr(E_{i:i+k-1}) = ((1-\theta_i)(1-\theta_d)(1-\theta_s))^k$
\end{corollary}
\begin{proof}
    This follows directly from the previous lemma: $\Pr(E_{i:i+k-1}) = \prod_{\ell = i}^{i+k-1} \Pr(E_\ell) = ((1-\theta_i)(1-\theta_d)(1-\theta_s))^k$.
\end{proof}

Applying the previous lemma to non-overlapping $k$-mers in $S$ yields a result that we sought to prove: the events that non-overlapping $k$-mer regions in $S$ have no mutations are independent.

\begin{corollary}[Independence of non-overlapping k-blocks]\label{cor:hom_indepdence}
    Let $E_{i:i+k-1}$ be the event that there are no mutations in $S[i:i+k-1]$ for $i \in [p+1, p+m']$. If $M$ is a collection of these $E_{i:i+k-1}$ such that no two overlap in S, then the variables in $M$ are independent. 
\end{corollary}
\begin{proof}
    Note that $\Pr(E_{i:i+k-1}) = \Pr(E_i,\dots, E_{i+k-1})$. If $E_{i:i+k-1}$ and $E_{j:j+k-1}$ do not share any $E_\ell$, they are independent.
    
    For any $i,j \in S$ such that $i \ne j$, then $|i - j| \ge k$. This implies that $E_{i:i+k-1}$ and $E_{j:j+k-1}$ do not overlap: they contain distinct random variables. Applying the previous lemma gives the result.
\end{proof}

\section{Anchor-count technical lemmas}
\begin{lemma}\label{convergence_bounds}
    For $t_0 = \frac{1}{2}\ln(\frac{9}{1 + 8\gamma})$, $e^{t_0}\rho_i' < 1$ and $(1 - \theta_i) + \theta_i(\frac{(1 - \rho_i')e^t}{1 - e^t\rho_i'}) \le e$ for all $0 < \rho_i' < \gamma$.
\end{lemma}
\begin{proof}
    From the choice of $t_0$, it follows that $e^{t_0} \le \frac{9}{1 + 8\gamma}$. Note that $f(x) = \frac{x}{1 + \gamma (x-1)}$ is increasing for $x > 0$. We have $9 < \frac{e - 1}{0.206} + 1 \le \frac{e - 1}{\theta_i} + 1$ because $\theta_i \le \theta_T < 0.206$. Letting  $A = \frac{e - 1}{\theta_i} + 1$, the previous two facts give that $e^{t_0} \le \frac{A}{1 + \gamma(A - 1)}$, and since $\rho_i' \le \gamma$, we have additionally that $e^{t_0}\le \frac{A}{1 + \gamma(A - 1)} \le \frac{A}{1 + \rho_i'(A - 1)}$.

    Multiplying both sides of this inequality by $\rho_i'$ gives $e^{t_0}\rho_i' \le \frac{A\rho_i'}{\rho_i'A + (1 - \rho_i')} < 1$ since $\rho_i' < \gamma < 1$. This final inequality shows why it is necessary to bound $\rho_i'$ away from 1.
    
    The second inequality in the lemma follows by rearrangement and inserting the expression represented by $A$:
    
    Expanding the inequality gives $(1 - \rho_i')e^{t_0} + e^{t_0}\rho_i'A \le A$, rearranging we get $(1 - \rho_i')e^{t_0} \le A(1 - \rho_i'e^{t_0})$. Simplifying yields $\frac{(1 - \rho_i')e^{t_0}}{1 - \rho_i'e^{t_0}} \le A$. Plugging in $A = \frac{e - 1}{\theta_i} + 1$ and multiplying both sides by $\theta_i$ gives $\theta_i(\frac{(1 - \rho_i')e^{t_0}}{1 - \rho_i'e^{t_0}}) \le e - 1 + \theta_i$, from which we immediately get $(1 - \theta_i) + \theta_i(\frac{(1 - \rho_i')e^t}{1 - e^t\rho_i'}) \le e$. This proves the second claim.
\end{proof}

\begin{lemma}[\textbf{Bounded Expansion ($\bm{E}$)}]\label{expansion-bound}
    With probability $\ge 1 - 1/n$, no $k$-mer in $S[p+1:p+m']$ has more than $\frac{1}{t_0}(\frac{2}{\beta} + 1)k$ inserted base pairs, for $t_0 = \frac{1}{2}\ln(\frac{9}{1 + 8\gamma})$.
\end{lemma}
\begin{proof}
    Denote the random variable representing the total insertion length at the $p+j$-th coordinate as $I_j$ where
\[
I_j = 
\begin{cases}
0, & \text{with probability }1 - \theta_i,\\
\mathrm{Geom}(1 - \rho_i'), & \text{with probability }\theta_i,
\end{cases}
\]

In particular, for $\ell > 0$, $\text{Pr}(I_j = \ell) = \theta_i (1-\rho_i') (\rho_i')^{\ell-1}$. Define $Z = \sum_{j=1}^k I_j$, which represents the total insertion length, the expansion, of the first $k$-block. A simple Chernoff bound shows that for any $t > 0$, $$\text{Pr}(Z \ge c) \le \frac{\mathbb{E}[e^{tZ}]}{e^{tc}} = \frac{\prod_{j=1}^k \mathbb{E}[e^{tI_j}]}{e^{tc}} = \frac{\mathbb{E}[e^{tI_1}]^k}{e^{tc}}.$$

Where the first equality follows since the $\{I_j\}_{j=1}^k$ are independent and the second from them being identically distributed. Choosing $t = t_0 = \frac{1}{2}\ln(\frac{9}{1 + 8\gamma})$, as in the previous lemma, we can calculate the moment generating function of $I_1$ can be directly:

$M_{I_1}(t_0) = \mathbb{E}[e^{t_0I_1}] = (1-\theta_i) + \sum_{j=1}^\infty \theta_i(1-\rho_i')(\rho_i')^{j-1}e^{t_0j} = (1-\theta_i) + \theta_i(1-\rho_i')e^{t_0}\sum_{j'=0}^\infty (\rho_i' e^{t_0})^{j'}$. The last term is a geometric series which converges since $\rho_i' e^t < 1$ by Lemma \ref{convergence_bounds}. Thus, $M_{I_1}(t) = 1 - \theta_i + \theta_i\frac{(1-\rho_i')e^{t_0}}{1-\rho_i'e^{t_0}}$, which the previous lemma (Lemma \ref{convergence_bounds}) shows is at most $e$.

Thus, $\text{Pr}(Z \ge c) \le e^{k - t_0c}$ for any $c \in \mathbb{R}$. Choosing $c = \frac{1}{t_0}(\frac{2}{\beta} + 1)k$ gives $\text{Pr}(Z_1 \ge \frac{1}{t_0}(\frac{2}{\beta} + 1)k) \le e^{k-(\frac{2}{\beta} + 1)k} = e^{-\frac{2}{\beta}k} = e^{-2C\ln(n)} \le e^{-2\text{ln}(n)} = \frac{1}{n^2}$ since $C > 1$. 

Define $Z_i$ to be the random variable denoting the expansion of the $S[p+i:p+i+k-1]$, the $i$-th block in $S[p+1:p+m']$, formally, $Z_i$ is the sum of the insertion lengths at each position in the $k$-mer $S[p+i:p+i+k-1]$. Note that each $Z_i$ has the same distribution as $Z$.

A simple union bound shows that $\text{Pr}(\exists j:Z_j \ge \frac{1}{t_0}(\frac{2}{\beta} + 1)k) \le (m'-k+1)\frac{1}{n^2} \le n\frac{1}{n^2} = \frac{1}{n}$, and the result follows.
\end{proof}

\begin{lemma}[\textbf{Bounded Contraction ($\bm{C}$)}]\label{contraction-bound}
    With probability $\ge 1 - 1/n$, no $\ell$-block in $S[p+1:p+m']$ contracts to size $\le \frac{(1 - \theta_d)\ell}{2}$, where $\ell = \frac{21k}{\beta}$.
\end{lemma}
\begin{proof}
    The proof follows similarly to the previous lemma. Define 
\[
X_j =
\begin{cases}
0, & \text{if the $j$‑th index of $S[p+1:p+m']$ is deleted (with probability $\theta_d$),}\\
1, & \text{otherwise.}
\end{cases}
\]

Note if $X_j = 1$, then the $j$-th index of $S[p+1:p+m']$ survives. Let $X = \sum_{j=1}^\ell X_j$ be the total number of surviving indices in the $\ell$-block and set $q_d = 1 - \theta_d$ to be the survival rate of an index. A classic Chernoff bound on the sum of i.i.d. Bernoulli random variables gives for any $0 < \delta < 1$:
$\text{Pr}(X \le (1 - \delta)q_d\ell) \le \text{exp}(-\frac{\delta^2 q_d\ell}{2})$. Since $\ell = ck = c\beta C\text{ln}(n)$ where $c = \frac{21}{\beta}$ gives $\frac{\delta^2}{2}q_dc_0C\beta \ge \frac{1}{8}(.794)21 > 2$, when $\delta = 1/2$ and $\theta_d \le\theta_T \le 0.206$. Thus, $\text{Pr}(X \le \frac{1}{2}q_d\ell) \le \text{exp}(-2\text{ln}(n)) = \frac{1}{n^2}$.

A union bound over all $\ell$-blocks in $S[p+1:p+m']$ gives: $$\text{Pr}(\exists \ell \text{-block shrunk to less than } \frac{(1 - \theta_d)\ell}{2}) \le n /n^2 = \frac{1}{n}.$$
\end{proof}

\begin{proof}[Proof of Lemma \ref{var_bound}]
    Let $S_p = \{(i, j) \in S \times S' \mid [(i,j), \dots, (i+k-1,j+k-1)] \,\ \cap P_H = \emptyset \}$ be the set of all positions where an anchor at that position does not intersect the homologous path.
    
    Define $B_k(i, j) = \{(h, l) \in S_p \mid |h - i| \le k \,\ \cap \,\ |l - j| \le k \}$, and

$$
\begin{aligned}
P_k(i,j)
=\;\bigl\{(h,l)\in S_p :\;&
   f^{-1}\bigl([l:l+k-1]\bigr)\,\cap\,[i:i+k-1]\neq\emptyset,\\
&\quad
   f^{-1}\bigl([j:j+k-1]\bigr)\,\cap\,[h:h+k-1]\neq\emptyset
\bigr\}.
\end{aligned}
$$
    
    By Lemma \ref{indel-independence}, for any $(h, l) \notin B_k(i,j)\,\ \cup \,\ P_k(i,j)$, $A(i,j)$ and $A(h,l)$ are independent. Then, $N_S = \sum_{(i, j) \in S_p} A(i, j)$. We calculate the variance as follows:

\begin{align*}
N_S^2
&= \sum_{(h,l)\in S_p}\sum_{(i,j)\in S_p}A(i,j)A(h,l) \\ 
&= 
\underbrace{%
  \sum_{(i,j)\in S_p}\sum_{(h,l)\in S_p\setminus\bigl(B_k(i,j)\cup P_k(i,j)\bigr)}
  A(i,j)A(h,l)
}_{S_1}
+
\underbrace{%
  \sum_{(i,j)\in S_p}\sum_{(h,l)\in S_p\cap B_k(i,j)}
  A(i,j)A(h,l)
}_{S_2} \\[-0.5ex]
&\quad+
\underbrace{%
  \sum_{(i,j)\in S_p}\sum_{(h,l)\in\bigl(S_p\cap P_k(i,j)\bigr)\setminus B_k(i,j)}
  A(i,j)A(h,l)
}_{S_3}.
\end{align*}

Dealing first with $S_1$: by the independence lemma (Lemma \ref{indel-independence}), $A(h,l), A(i,j)$ are independent, so: 

\[
\begin{aligned}
\mathbb{E}(S_1)
&= 
\sum_{(i,j)\in S_p}\;\sum_{(h,l)\in S_p\setminus (B_k(i,j)\cup P_k(i,j))}
\mathbb{E}\bigl(A(i,j)\bigr)\,\mathbb{E}\bigl(A(h,l)\bigr) \\[0.75ex]
&\le 
\sum_{(h,l)\in S_p}
\;\sum_{(i,j)\in S_p}
\mathbb{E}\bigl(A(i,j)\bigr)\,\mathbb{E}\bigl(A(h,l)\bigr) = \mathbb{E}(N_S)^2.
\end{aligned}
\]

Note $\mathbb{E}(A(h,l)A(i,j)) \le \mathbb{E}(A(i,j))$ since the anchor random variables take values in $\{0,1\}$. Using that $|B_k(i,j)| \le 4k^2$, $\mathrm{Pr}(A(i,j) = 1) = \frac{1}{\sigma^k}$ for $(i,j) \in S_p$ from Corollary \ref{cor:anchor-prob}, and the naive bound $|S_p| \le mn$, we get:
$$\mathbb{E}(S_2) \le \sum_{(i,j)\in S_p}\sum_{(h,l)\in S_p\cap B_k(i,j)}
  \mathbb{E}(A(i,j)) \le 4k^2\frac{mn}{\sigma^k}.$$

Lastly, we handle the $S_3$ term. Working under $EC$, the region $S[i:i+k-1]$ can expand to have at most $(\frac{1}{t_0}(\frac{2}{\beta} + 1) + 1)k + 2k$ corresponding positions $l$ on S'. Note the additional $2k$ comes from $k$-mers that start before and after the corresponding region but still intersect it. The same argument shows that there are at most $(\frac{1}{t_0}(\frac{2}{\beta} + 1) + 1)k + 2k$ positions $j$ that correspond to $S[h:h+k-1]$. Thus, $|P_k(i,j)| \le ((\frac{2}{\beta} + 1) + 1)k + 2k)^2 = (\frac{1}{t_0}(\frac{2}{\beta} + 1) + 3)^2k^2 \le \frac{1}{2}T_0 k^2$.


This yields,

$$\mathbb{E}(N_S^2) \le \mathbb{E}(N_S)^2 + T_0k^2\frac{mn}{\sigma^k}.$$

From which it immediately follows that $\mathrm{var}(N_S) \le T_0k^2\frac{mn}{\sigma^k}$.
\end{proof}

\begin{proof}[Proof of Lemma \ref{lemma: spurious_anchor_count}]
    Call the event $X = \{N_S \ge n^{2 - C} + \sqrt{T_0\,}\,C\log(n)\,n^{\frac{3-C}{2}}\}$, i.e., the event that there are more spurious anchors than the amount given by the expression. We are looking to bound $\Pr\bigl(X \mid EC\bigr)$

    To do this, let us first show that $\frac{mn}{\sigma^k} \le n^{2-C}$. Recall that $k = C\log(n)$. Under $EC$, $m \le cm'$ for the expansion constant $c > 0$. Since $m' \ll n$, it follows that $m \ll n$. Thus, $\log(m) \le \log(n)$. From this inequality, we get immediately that $(C-1)\log(n) + \log(m) \le C\log(n) = k$. Subtracting $(C-1)\log(n)$ and adding $\log n$ to both sides yields $\log(mn) - k \le (2-C)\log(n)$. Finally, $\log(\frac{mn}{\sigma^k}) \le \log(n^{2-C})$, which gives the inequality. 

    Now, by Chebyshev, $\text{Pr}(N_S \ge \mathbb{E}(N_S \mid EC) + \sqrt{n\mathrm{var}(N_S \mid EC)} \mid EC) \le \frac{1}{n}$. We can upper bound $\mathrm{Pr}(X \mid EC)$ by using the variance bound of $N_S$ given by Lemma \ref{var_bound} and the previous inequality. This yields: $\mathrm{var}(N_S \mid EC) \le T_0k^2\frac{mn}{\sigma^k}$ , so $\frac{1}{n} \ge \mathrm{Pr}(N_S \ge \mathbb{E}(N_S \mid EC) + \sqrt{n\mathrm{var}(N_S \mid EC)} \mid EC) \ge \mathrm{Pr}(N_S \ge \mathbb{E}(N_S \mid EC) + \sqrt{T_0}C\log(n)n^{\frac{3-C}{2}} \mid EC)$. Finally, bounding $\mathbb{E}(N_S \mid EC) \le \frac{mn}{\sigma^k}$ and using our previously shown inequality, we obtain:
    
    $$\Pr\Bigl(N_S \ge n^{2 - C} + \sqrt{T_0\,}\,C\log(n)\,n^{\frac{3-C}{2}} \mid EC\Bigr) \le \frac{1}{n}.$$ 
\end{proof}

\section{Runtime theorem proof}
\label{sec:runtime-thm-proof}

The runtime of seed chain extend $T_{SCE}$ is the runtime of chaining $T_{Chain}$ and extension $T_{Ext}$. Since we are using linear gap costs, the runtime of chaining will be $O(N\log N)$. To show that $\mathbb{E}(T_{SCE}) = O(mn^{C\alpha}\log n)$ for any optimal chain, we will first show that $\mathbb{E}(T_{Chain}) = O(mn^{C\alpha}\log n)$.

\begin{lemma}\label{lemma:chaining_runtime}
    The expected runtime of chaining $\mathbb{E}(T_{Chain})$ for any optimal chain is $O(mn^{C\alpha}\log n)$.
\end{lemma}
\begin{proof}
    Under $G$, all anchors are homologous or clipping. We will obtain a loose bound on the total number of anchors under $G$. Note first that there can be at most $m$ homologous anchors since $|S'| = m$. Each point on $|P_H|$ can be in at most $k$ clipping anchors: it can be in any one of the $k$ positions in a clipping anchor. Thus, there can be at most $k|P_H|$ clipping anchors. Under $G$, $|P_H| \le c|S'| = cm$ for some constant $c > 0$. Thus, $N \le m + ckm$. 

    It follows that $N\log N \le (m + ckm)\log (m + ckm)$, and so $\mathbb{E}(N\log N \mid G) \le (m + ckm)\log (m + ckm)$. Since $\log(m+ckm) \ll \log(n)$ and $ck \ll n^{C\alpha}$, it follows that $\mathbb{E}(N\log N \mid G) = O(mn^{C\alpha}\log n)$.

    We use a worst-case bound for the case where $G^c$ occurs. Note that $N \le mn$ since there can be at most one anchor for each pair of positions in $S,S'$. Thus, $\mathbb{E}(N\log N \mid G^c)\Pr(G^c) \le mn \log (mn)\frac{6}{n} = O(m\log(mn)) = O(m\log(n))$. 

    By the law of total expectation, $\mathbb{E}(N\log N) = \mathbb{E}(N\log N \mid G)\Pr(G) + \mathbb{E}(N\log N \mid G^c)\Pr(G^c) \le O(mn^{C\alpha}\log n) + O(m\log(n)) =O(mn^{C\alpha}\log n)$.
\end{proof}

The strategy for showing that $\mathbb{E}(T_{Ext}) = O(mn^{C\alpha}\log n)$ for an optimal chain will be to upper bound its extension runtime by that of the subchain of homologous anchors without any mutations occurring at ``nice'' positions: $\{p+1,p+1+k,p+1+2k,\dots\}$. Intuitively, this works because that subchain is sparser and extension through sparser subchains takes longer. We begin by showing that every mutation-free homologous anchor that lies between the first and last anchor in an optimal chain must belong to it.

\begin{lemma}[Working in $\bm{G}$]\label{lemma:K_subset_C}
    Let $\mathcal{C} = ((i_1,j_1),\dots,(i_u,j_u))$ be an optimal chain. Under $G$, it consists only of homologous and clipping anchors. Then if there is an index $a \in [p+1, p + m' - k + 1]$ such that $i_1 \le a \le i_u$ and there are no mutations in $\{a,\dots,a+k-1\}$, then $(a, f(a)) \in C$.
\end{lemma}
\begin{proof}
    If $(a,f(a)) \in \mathcal{C}$, then the result is trivially true. Assume that $(a,f(a)) \notin \mathcal{C}$. We must show that adding it to $\mathcal{C}$ is both valid and increases the score of $\mathcal{C}$.

    Turning to the first part: if adding $(a,f(a))$ breaks the anchor monotonicity property of the chain, then there exists some anchor $(i,j)$ such that $(i,j)$ and $(a,f(a))$ intersect the homologous path and either $i \le a$ and $j \ge f(a)$ or $i \ge a$ and $j \le f(a)$. We consider the first case, $i \le a$ and $j \ge b$, noting that the proof for the second case is identical. Since $(a,f(a))$ is homologous, all of its points lie on $P_H$. Furthermore, since there no mutations occur in the $k$ positions $[a:a+k-1]$ on $S$, it follows that the points $(a+t,f(a)+t) \in P_H$ for $0 \le t \le k$. Since $P_H$ touches the anchor $(i,j)$, this means that there is a portion of $P_H$ that connects $(a+k,f(a) + k)$ to some point $(x,y) \in \{(i+t,j+t)\mid 0 \le t \le k - 1\}$. However, since $a \ge i$, it follows that $a + k > i + k - 1$, a contradiction. It follows that $(a,f(a))$ can be added to $\mathcal{C}$.

    Lastly, adding $(a,f(a))$ to $\mathcal{C}$ increases its score by $1$, since there is no change in the linear-gap cost, but there is an additional anchor. 
\end{proof}

The lemma above applies specifically to preserved homologous anchors; clipping anchors can be mutually exclusive. In particular, the homologous path can ``weave'' between the two clipping anchors with the same $x$ value, implying that both cannot be in the same chain. 

The lemma below establishes that if two chains share a first and last anchor and one is a subchain of the other, then its extension runtime is greater.

\begin{lemma}[Extension takes longer for sparser chains]\label{lemma:sparser_chains_take_longer}
    Let $\mathcal{C} = ((i_1,j_1), \dots, (i_u, j_u))$ and $\mathcal{C}' = ((i_1',j_1'), \dots, (i_u', j_u'))$ be two chains such that $(i_1,j_1) = (i_1',j_1')$ and $(i_u, j_u) = (i_u', j_u')$ and $\mathcal{C}' \subset \mathcal{C}$. Then the runtime of extension through $\mathcal{C}'$ is $\ge$ the runtime of extension through $\mathcal{C}$.
\end{lemma}
\begin{proof}
    The trivial cases are when $|\mathcal{C}| = |\mathcal{C}'| = 1$ or $|\mathcal{C}| = |\mathcal{C}'| = 2$, in which case they are the same chain and the result holds trivially. For the remainder of the proof, assume that $|\mathcal{C}| > 2$.

    Denote the runtime of extension in the subchain $(i_a,j_a)$ to $(i_b,j_b)$ of a chain $\mathcal{C}_0$ as $T_{Ext}((i_a,j_a):(i_b,j_b);\mathcal{C}_0)$.  Then we can express the extension runtime of $\mathcal{C}$ as $T_{Ext}(\mathcal{C}) = \sum_{\ell=1}^{u'-1}T_{Ext}((i_{\ell}',j_{\ell}'):(i_{\ell + 1}',j_{\ell + 1}'); \mathcal{C})$. In other words, the extension runtime of $\mathcal{C}$ can be decomposed as the sum of the extension runtimes of the subchains induced by $\mathcal{C}'$ since they share the same endpoints.

    Since $T_{Ext}(\mathcal{C}') = \sum_{\ell=1}^{u'-1}T_{Ext}((i_{\ell}',j_{\ell}'):(i_{\ell + 1}',j_{\ell + 1}'); \mathcal{C}')$, it suffices to show that $T_{Ext}((i_{\ell}',j_{\ell}'):(i_{\ell + 1}',j_{\ell + 1}'); \mathcal{C}) \le T_{Ext}((i_{\ell}',j_{\ell}'):(i_{\ell + 1}',j_{\ell + 1}'); \mathcal{C}')$ for each $1 \le \ell \le u' - 1$. 

    Without loss of generality, let $\ell = 1$. Denote the subchain of $S$ between $(i_1', j_1'), (i_2', j_2')$ as $((i_1,j_1),\dots (i_p,j_p))$ where $(i_p,j_p) = (i_2', j_2')$. Furthermore, let $G_\ell(S;\mathcal{C}) = \max (i_{\ell +1} - i_{\ell} - k + 1, 0), G_\ell(S';\mathcal{C}) = \max (j_{\ell +1} - j_{\ell} - k + 1, 0)$ be the gap size of the $\ell$-th anchor on $S$ and $S'$. The values $G_\ell(S;\mathcal{C}')$ and $G_\ell(S';\mathcal{C}')$ are defined similarly. Since $i_2' - i_1' = \sum_{\ell = 1}^{p-1} i_{\ell+1} - i_{\ell}$, it is clear that if $i_2' - i_1' - k + 1 \le 0$, then $G_\ell(S;\mathcal{C}) = 0$ for all $1 \le \ell \le p - 1$. Otherwise, it is also clear that $G_{1}(S;\mathcal{C}') = i_2' - i_1' - k + 1 \ge \sum_{\ell = 1}^{p-1} G_{\ell}(S;\mathcal{C})$ and similarly, $G_{1}(S';\mathcal{C}')  \ge G_1(S';\mathcal{C})$.

    We have that $T_{Ext}((i_1,j_1):(i_p,j_p);\mathcal{C}) = \sum_{\ell = 1}^{p-1} G_{\ell}(S;\mathcal{C})G_{\ell}(S';\mathcal{C})$. Since each gap is non-negative, this is $\le( \sum_{\ell = 1}^{p-1} G_{\ell}(S;\mathcal{C}))(\sum_{\ell = 1}^{p-1} G_{\ell}(S';\mathcal{C})) \le G_{1}(S;\mathcal{C}')G_{1}(S';\mathcal{C}')$ which is exactly $T_{Ext}((i_1',j_1'):(i_2',j_2');\mathcal{C}')$. Thus, we have shown that $T_{Ext}((i_1,j_1):(i_p,j_p);\mathcal{C}) \le T_{Ext}((i_1',j_1'):(i_2',j_2');\mathcal{C}')$, as desired.
    
\end{proof}

\begin{definition}
    Recall that a preserved anchor is some anchor $(i,f(i))$  for which there are no mutations at any position in the $k$-mer $S[i:i+k-1]$. Define $K$ to be the chain of all preserved anchors occurring at $(i,f(i))$ for $i \in \{p+1,p+1+k,\dots,p+1+(\lfloor \frac{S}{k} \rfloor - 1)k\}$.

    Define $Y_i^K$ to be the random variable denoting the number of uncovered bases between a preserved anchor occurring at position $i \in \{p+1,p+1+k,\dots,p+1+(\lfloor \frac{S}{k} \rfloor - 1)k\}$, if it exists, and the next preserved anchor in $K$. Otherwise, $Y_i^K = 0$. Formally, $Y_i^K = \ell$ if positions $[i,\dots, i+k-1]$ and $[i+\ell +k,\dots, i+\ell + 2k - 1]$ on $S$ are unmutated and there is no preserved anchor in between. If $(i,f(i))$ is not a preserved anchor, then $Y_i^K = 0$.
\end{definition}

With the tools developed in the previous two lemmas, we now prove that the extension through the convenient anchors $K$, with a particular start and end anchor chosen, takes longer than extending through $\mathcal{C} \supset K$.

\begin{lemma}[Working in $\bm{G}$]\label{lemma:extension_bound}
    Let $\mathcal{C}$ be any optimal chain for $S,S'$. Define $G_{start}(S), G_{start}(S')$,  to be the distance from the start of $P_H$ to the first anchor of $K$ in $S$ and $S'$, respectively, and if $K$ is empty, then $G_{start}(S) = m', G_{start}(S') = m$. $G_{end}(S)$ and $G_{end}(S')$ are defined similarly. 

    Then $T_{ext}(\mathcal{C}) \le (G_{start}(S) + k)(G_{start}(S')) + T_{ext}(K) + (G_{end}(S) + k)(G_{end}(S'))$
\end{lemma}
\begin{proof}
    Let $\mathcal{C} = ((i_1, j_1), \dots, (i_u, j_u))$. Since we are working in $G$, all of these anchors are homologous or clipping. It follows that $i_1 \ge p - k + 1$, since otherwise $\{(i_1+t, j_1+t) \mid 0 \le t \le k - 1\} \cap P_H = \emptyset$ since $i_1 + k - 1 \le p - k + 1 + k - 1 = p < p + 1$.

    There are two cases to consider. The first is when $K$ is empty, i.e., there are no anchors occurring at positions $p+1+kt$ in $S[p+1:p+m']$ for $0 \le t$. In this case, $G_{start}(S) = m', G_{start}(S') = m$, and the expression on the right hand side evaluates to $2(m+k)(m'+k)$. Since $i_1 \ge p - k + 1, j_1 \ge 0$ and, similarly, $i_u \le p + m' + k, j_u \le m$, it follows that extending between $(i_1,j_1)$ and $(i_u, j_u)$ has runtime at most $(i_u - i_1)(j_u - j_1) \le (m + 2k)(m') \le 2(m+k)(m') = 2(G_{start}(S) + k)(G_{start}(S'))$ since $m,m' \gg k$ under $G$. This establishes the first case.

    Assume now that $K$ is non-empty. Let the first anchor in $K$ be $(i_p, j_p)$ and the last anchor be $(i_q, j_q)$. If $|K| = 1$, then they are the same. Define $T_{Ext}(K;\mathcal{C})$ to be the runtime of extension for anchors occurring between $(i_p, j_p)$ and $(i_q, j_q)$ in $\mathcal{C}$. This is well-defined, and $K \subset \mathcal{C}$ by Lemma \ref{lemma:K_subset_C}. Then by Lemma \ref{lemma:sparser_chains_take_longer}, $T_{Ext}(K;\mathcal{C}) \le T_{Ext}(K)$. By Lemma \ref{lemma:sparser_chains_take_longer}, we also have that the extension runtime of the chain $\mathcal{C}' = ((i_1, j_1), (i_p, j_p))$ is at least that of the sub-chain containing all anchors from $(i_1, j_1)$ up to and including $(i_p, j_p)$; note that if $(i_1,j_1) = (i_p, j_p)$, then $\mathcal{C}'$ is technically not a valid chain but then the extension runtimes mentioned are both $0$. Lastly, since $i_1 \ge p + 1 - k$, it follows that the extension region defined by the endpoints $(p+1-k, 0), (i_p, j_p)$ contains the region defined by endpoints $(i_1, j_1), (i_p, j_p)$, and so the extension runtime of the first box is at least $T_{Ext}(\mathcal{C}')$. The extension of the first region is exactly $(i_p - (p + 1) + k)(j_p) = (G_{start}(S) + k)(G_{start}(S'))$. By the same argument applied to the extension after $(i_q, j_q)$, we obtain the desired inequality.
\end{proof}

The only remaining task is to bound $\mathbb{E}(T_{Ext})$. This is done below by making use of the previous lemma and the law of total expectation with respect to the space $G$.

\begin{lemma}\label{lemma:extension_runtime}
    The expected runtime of extension  $\mathbb{E}[T_{Ext}]$ for any optimal chain $\mathcal{C}$ is $O(mkn^{C\alpha})$.
\end{lemma}
\begin{proof}
    By linearity of expectation, $\mathbb{E}[T_{Ext}(\mathcal{C})] = \mathbb{E}[T_{Ext}(\mathcal{C}) \mid G]\Pr(G) + \mathbb{E}[T_{Ext}(\mathcal{C}) \mid G^c]\Pr(G^c)$. Using that $\Pr(G^c) \le \frac{6}{n}$, $\mathbb{E}(T_{Ext}(\mathcal{C}) \mid G^c)\Pr(G^c) \le mn(\frac{6}{n}) = O(m)$ since the runtime extension of $\mathcal{C}$ cannot exceed the runtime of a full naive alignment between $S$ and $S'$. It suffices to prove then that $\mathbb{E}[T_{Ext}(\mathcal{C}) \mid G]\Pr(G) = O(mkn^{C\alpha})$. 
    
    The general strategy to bound conditional expectations will be to first use the $G$ space to rewrite the expectations, then to bound the unconditional expectations that have convenient probability independence properties, and then to convert that to a bound on the conditional expectation.

    In Lemma \ref{lemma:extension_bound}, we showed that $\mathbb{E}(T_{Ext}(\mathcal{C}) \mid G) \le \mathbb{E}[(G_{start}(S) + k)(G_{start}(S')) \mid G] + \mathbb{E}[T_{ext}(K) \mid G] + \mathbb{E}[(G_{end}(S) + k)(G_{end}(S')) \mid G]$. We show that each term is $O(mkn^{C\alpha})$, from which it follows that the sum is also $O(mkn^{C\alpha})$.

    First, consider $\mathbb{E}[(G_{start}(S) + k)(G_{start}(S')) \mid G]$. Under $G$, $G_{start}(S') \le c(G_{start} + k)$ by Lemma \ref{lemma:EC_space}, where $c$ is the expansion constant; the number of points corresponding to the region in $S$ before the first $k$-mer without a mutation cannot expand more than $c$ times. The addition of $k$ is to ensure that $G_{start}(S) + k \ge k$, which is necessary since Lemma \ref{lemma:EC_space} applies to $k$-mers. Thus, under $G$, $(G_{start}(S) + k)(G_{start}(S')) \le c(G_{start}(S) + k)^2$. We now work with the unconditional expectation $\mathbb{E}[c(G_{start}(S) + k)^2] = c(\mathbb{E}[G_{start}(S)^2] + 2k\mathbb{E}[G_{start}(S)] + k^2)$. Note that $G_{start} = \ell k$ for $\ell \ge 0$. We can check that $kG_{start}(S) \le G_{start}(S)^2$. If $G_{start} = 0$, then both sides vanish, and otherwise $kG_{start}(S) = \ell k^2 \le (\ell k)^2 = G_{start}(S)^2$. Thus, $2k\mathbb{E}[G_{start}(S)] = O(\mathbb{E}[G_{start}(S)^2])$. We calculate this term now. Note that $\Pr(G_{start}(S) = \ell k) \le (1 - ((1 - \theta_i)(1 - \theta_d)(1-\theta_s))^{k})^{\ell} ((1 - \theta_i)(1 - \theta_d)(1-\theta_s))^{k}$. Writing $\theta_0  = (1 - \theta_i)(1 - \theta_d)(1-\theta_s) \ge 1 - \theta_T$, is the probability that no mutation occurs at a position. With this simpler notation, $\Pr(G_{start}(S) = \ell k) \le (1 - \theta_0^k)^\ell \theta_0^k$. It follows that 

    $$
    \mathbb{E}[G_{start}(S)^2] \le \theta_0^k \sum_{\ell = 1}^\infty (\ell k)^2 (1 - \theta_0^k)^\ell = O(\frac{k^2 \theta_0^k}{\theta_0^{3k}}) = O(n^{2C\alpha}\log^2 n)
    $$ since $\theta_0 \ge 1 - \theta_T$, and using the fact that for $|x| < 1$, $\sum_{i=1}^{\infty} i^2 x^i = \frac{x(x + 1)}{(1 - x)^3}$. Since $m' = \Omega(n^{2C\alpha + \epsilon})$ and $k = O(\log n)$, it follows that $n^{2C\alpha}\log^2 n \ll m'$, so $\mathbb{E}[G_{start}(S)^2] = O(m')$. Since $k^2 = O(m')$, it follows that $\mathbb{E}[c(G_{start}(S) + k)^2] = O(m')$. By linearity of expectation, $\mathbb{E}[(G_{start}(S) + k)(G_{start}(S')) \mid G] \le \frac{\mathbb{E}[c(G_{start}(S) + k)^2]}{\Pr(G)} \le \frac{m'}{1 - \frac{6}{n}} \le \frac{cm}{1 - \frac{6}{n}}$, where we used that $m' \le cm$ under $G$ by Lemma \ref{lemma:EC_space}. Thus, $\mathbb{E}[(G_{start}(S) + k)(G_{start}(S')) \mid G] = O(m)$, and by the same argument, $\mathbb{E}[(G_{end}(S) + k)(G_{end}(S')) \mid G] = O(m)$.

    We will now bound the term $\mathbb{E}[T_{ext}(K) \mid G]$. Under $G$, the corresponding distance in $S'$ from position $i$ to the start of the next $k$-mer without mutations is at most $c(Y_i^K + k)$ where $c$ is the expansion factor. The addition of $k$ comes from the fact that $Y_i^K$ can be less than $k$ and the $EC$ lemma bounds the expansion of $k$-mer regions of $S$ while generating $S'$. Thus, $\mathbb{E}[T_{ext}(K) \mid G] \le \sum_{i=p+1}^{p+m'}\mathbb{E}[cY_i^K(Y_i^K + k) \mid G] = \sum_{i=p+1}^{p+m'}c\mathbb{E}[(Y_i^K)^2 \mid G] + \sum_{i=p+1}^{p+m'}k\mathbb{E}[Y_i^K\mid G]$. We will bound both of these terms separately and show that they are both $O(mkn^{C\alpha})$.

    We first bound the unconditional version of the first term: $\sum_{i=p+1}^{p+m'}\mathbb{E}[(Y_i^K)^2]$. If $Y_i^K = \ell k$, then there are two regions of $k$ consecutive positions without any mutations, and there are $\ell k$ positions for which each $k$ block starting at position $p + 1 + kt$ in $S$ has some mutation. Thus, $\Pr(Y_i^K = \ell k) \le (1 - \theta_i)(1 - \theta_d)(1-\theta_s)^{2k} (1 - ((1 - \theta_i)(1 - \theta_d)(1-\theta_s))^{k})^{\ell}$. Thus, $\Pr(Y_i^K = \ell k) \le \theta_0^{2k}(1 - \theta_0^k)^{\ell}$. We get that $\sum_{i=p+1}^{p+m'}\mathbb{E}[(Y_i^K)^2] \le  \frac{m'}{k}\sum_{\ell=1}^\infty (\ell k)^2 \theta_0^{2k}(1 - \theta_0^k)^{\ell}$ since there are at most $\frac{m'}{k}$ non-zero random variables $Y_i^K$ for $i \in \{p+1, p+1+k, \dots,\}$. 

    Again, using the fact that for $|x| < 1$, $\sum_{i=1}^{\infty} i^2 x^i = \frac{x(x + 1)}{(1 - x)^3}$, we get that $\frac{m'}{k}\sum_{\ell=1}^\infty (\ell k)^2 \theta_0^{2k}(1 - \theta_0^k)^{\ell} = \frac{m'}{k}\theta_0^{2k}k^2\frac{(1-\theta_0^k)(2-\theta_0^k)}{\theta_0^{3k}} \le \frac{m'k}{\theta_0^k}$. Using that $\theta_0 \ge 1 - \theta_T$, we obtain $\sum_{i=p+1}^{p+m'}\mathbb{E}[(Y_i^K)^2] \le m'k(1-\theta_T)^{-k} = m'kn^{C\alpha}$. 

    Now, we will bound the second unconditional term: $\sum_{i=p+1}^{p+m'}\mathbb{E}[kY_i^K]$. If $Y_i^K = 0$, then $kY_i^K = (Y_i^K)^2 = 0$, and otherwise, $Y_i^K > k$. Thus, $kY_i^K \le (Y_i^K)^2$. It follows that $\sum_{i=p+1}^{p+m'}\mathbb{E}[kY_i^K] \le \sum_{i=p+1}^{p+m'}\mathbb{E}[(Y_i^K)^2] \le m'kn^{C\alpha}$.

    Combining the two inequalities and using linearity of expectation yields that $\mathbb{E}(T_{Ext} \mid G) \le \sum_{i=p+1}^{p+m'}\mathbb{E}[cY_i^K(Y_i^K + k) \mid G] \le \frac{m'kn^{C\alpha}}{\Pr(G)} \le \frac{m'kn^{C\alpha}}{1 - \frac{6}{n}} \le \frac{cmkn^{C\alpha}}{1 - \frac{6}{n}} = O(mkn^{C\alpha})$ since $m' \le cm$ for the expansion constant $c$ under $G$ by Lemma \ref{lemma:EC_space}. 

    Note that $mkn^{C\alpha} \gg m$, so we have shown that all three terms in the upper bound of $\mathbb{E}[T_{Ext}(\mathcal{C}) \mid G]$ are $O(mkn^{C\alpha})$. We conclude then that $\mathbb{E}[T_{Ext}(\mathcal{C}) \mid G] = O(mkn^{C\alpha})$, which completes the proof.
    
\end{proof}

Having calculated $\mathbb{E}(T_{Chain})$ and $\mathbb{E}(T_{Ext})$ for an optimal chain, we immediately obtain the full runtime result below.

\begin{theorem}
    The expected runtime $\mathbb{E}[T_{SCE}]$ of seed chain extend for any optimal chain $C$ under the constraints given in Def.~\ref{sce_defs} is $O(mn^{C\alpha}\log n)$.
\end{theorem}
\begin{proof}
    The runtime $T_{SCE}$ of any chain $\mathcal{C}$ is $T_{Chain} + T_{Ext}$, the runtime of seeding the query $S'$, chaining and the runtime of extension. Seeding $S'$ is fast and  takes $O(m)$ time. 

    Since $\mathcal{C}$ is an optimal chain, by Lemma \ref{lemma:chaining_runtime}, $\mathbb{E}[T_{Chain}] = O(mn^{C\alpha}\log n)$, and by Lemma \ref{lemma:extension_runtime}, $\mathbb{E}[T_{Ext}] = O(mn^{C\alpha}\log n)$. Thus, $\mathbb{E}[T_{SCE}] = O(m) + O(\mathbb{E}[T_{Chain}]) + O(\mathbb{E}[T_{Ext}]) = O(mn^{C\alpha}\log n)$.
\end{proof}

\section{Old recoverability and clipping anchors}
\label{sec:old-recoverability}

The main difference between the analysis presented in this paper, where indels are present, and the prequel, which deals with the substitution-only model, is the presence of clipping anchors. The purpose of this section is to briefly describe some of the issues surfaced by clipping anchors, and the motivation of our recoverability definition to solve these issues. An overview of the issues will be given first, and proofs will follow.

Recall that a clipping anchor is an anchor that contains a point on the homologous path but does not entirely ``contain'' the path between the start and end of the anchor. Formally, if the anchor begins at $(i,j) \in |S| \times |S'|$, then it is clipping if $\{(i+t,j+t) \mid 0 \le t \le k-1 \} = \{(x,y) \in P_H \mid i \le x \le i+k-1 \land j \le y \le j+k-1\}$. The primary difficulty from clipping anchors is that for large enough $n$, they are on the same order as the number of homologous anchors. In the setup shown in the paper, any optimal chain is all homologous and clipping anchors. Because of this, the expected number of points missed due to clipping anchors is at least $O(\mathbb{E}(N_C))$. However, $\mathbb{E}(N_C) = O(\frac{m}{k}n^{-C\alpha})$. 

The initial recoverability definition we used was a simple extension of the one in the prequel: all points on the homologous path that lie in an anchor or an extension region are considered ``recovered'', as shown below.

\begin{definition}[Initial definition of recoverability]
    Given a chain $\mathcal{C} = ((i_1,j_1),\dots,(i_u,j_u))$, we define the union of all possible alignments for the chain $\mathcal{C}$, $Align(\mathcal{C})$, as:
    $$Align(\mathcal{C}) = \bigcup_{\ell=1}^u \{ (i_\ell,j_\ell),\dots, (i_\ell +k-1, j_\ell + k - 1)  \} \,\ \cup \,\ \bigcup_{\ell=1}^{u-1} Ext(\ell).$$ Where $Ext(\ell) = \{i_\ell + k - 1,\dots,i_{\ell+1}\} \times \{j_\ell + k - 1,\dots,j_{\ell+1}\}$. If $i_\ell + k - 1> i_{\ell + 1}$ or $j_\ell + k - 1 > j_{\ell + 1}$, then $Ext(\ell) = \emptyset$.

    The recoverability of the chain, $R(\mathcal{C})$, is defined to be:
    $$R(\mathcal{C}) = \frac{|Align(\mathcal{C}) \cap P_H|}{|P_H|}$$
    
\end{definition}

Using this definition, we can bound the expected recoverability of any optimal chain as follows: $\mathbb{E}(R) = \mathbb{E}(R \mid G)\Pr(G) + \mathbb{E}(R \mid G^c)\Pr(G^c) \le \mathbb{E}(R \mid G) + \Pr(G^c) \le \mathbb{E}(R \mid G) + \frac{6}{n}$. In the $G$ space, any optimal chain will contain all homologous anchors and, in the worst case, all clipping anchors. When this occurs, $\mathbb{E}(R \mid G) \le 1 - \mathbb{E}(\frac{N_C}{k|P_H|})$. In other words, the number of missed points is at least the number of clipping anchors over $k$ since there is at least one missed point on $P_H$ for each clipping anchor and at most $k$ clipping anchors can share any point. This is a vast undercount but suffices to show the difficulty of clipping anchors.

We will now lower bound $\mathbb{E}(\frac{N_C}{k|P_H|} \mid G)$. Under $G$, $|P_H| = cm$ for some $c > 0$ due to bounded expansion and contraction (Lemma \ref{lemma:EC_space}). Furthermore, $\mathbb{E}(N_C \mid G) \ge \frac{1}{\sigma}\theta_i \theta_d (1 - \rho_i')\frac{m}{k}(1 - \theta_T)^{k-1} = \frac{\theta_i \theta_d}{\sigma(1-\theta_T)} (1 - \rho_i')\frac{m}{k}(1 - \theta_T)^{k} $. We can rewrite this as $c'\frac{m}{k}(1-\theta_T)^{k}$ after absorbing constants into $c'$. This bound will also be shown rigorously, but the idea is that a clipping anchor can occur from having $k - 1$ positions with no mutations and then the last position in the anchor has an insertion and deletion. Since there are many more ways for clipping anchors to occur, this is clearly a lower bound.

Returning to the recoverability bound, we have that $\mathbb{E}(\frac{N_C}{k|P_H|} \mid G) \ge \frac{c'\frac{m}{k}(1-\theta_T)^{k}}{ckm} = \frac{c'(1 - \theta_T)^k}{k^2c} = O(\frac{1}{k^2}n^{-C\alpha})$. Thus, $\mathbb{E}(R) \le 1 - O(\frac{1}{k^2}n^{-C\alpha}) + \frac{6}{n}$

Note that since $m = \Omega(n^{2C\alpha + \epsilon})$, we have $m^{-1/2} = O(n^{-C\alpha - \epsilon/2}) = o(n^{-C\alpha})$, and so $\mathbb{E}(R \mid G) = 1 - O(\frac{n^{-C\alpha}}{k^2}) = 1 - \Omega(\frac{1}{\sqrt{m}})$. From this, we get that $\mathbb{E}(R) = 1 - \Omega(\frac{1}{\sqrt{m}})$. Thus, if we do not handle clipping anchors explicitly in the definition of recoverability, the expected recoverability converges more slowly than $1 - O(\frac{1}{\sqrt{m}})$.

Below we prove the statements used in the informal analysis above. Recall that $G = EC\land F1 \land F2$


\begin{lemma}[Working in $\bm{G}$]
    The expected number of clipping anchors is at least $O(mk(1 - \theta_T)^k)$
\end{lemma}
\begin{proof}
    Partition $S[p+1:p+m']$ into $\frac{m'}{k}$ blocks of length $k$. For each $k$-mer, there can be a clipping anchor produced by a simple ``kink'', i.e., an insertion at position $i$ of $S[i]$ followed by a deletion of $S[i]$. If there are two ``kinks'', the result is an anchor since the substrings on $S$ and $S'$ are left unchanged and that anchor clearly clips $P_H$ due to the kinks.

    There are $\binom{k}{2}$ ways to select the locations for the two ``kinks'' and the probability that they occur is $(\theta_i \theta_d (1 - \rho_i')\frac{1}{\sigma})^2 (1 - \theta_T)^{k-2} = c(1 - \theta_T)^k$ after absorbing constants and writing $(1 - \theta_T)^{k-2} = \frac{(1 - \theta_T)^{k}}{(1 - \theta_T)^2}$.
    
    Thus, the expected number of clipping anchors is at least $c\frac{m'}{k}\binom{k}{2}(1 - \theta_T)^{k} = O(m'k(1 - \theta_T)^k)$ after summing over all $\frac{m'}{k}$ $k$-mers in $S[p+1:p+m']$. Using the fact that $m' = cm$ under $G$ gives the result.
\end{proof}

\begin{lemma}[Working in $\bm{G}$]
    The expected number of homologous anchors is at most $O(m(1 - \theta_T)^k)$.
\end{lemma}
\begin{proof}
    For a homologous anchor to occur, there must be a sequence of at least $k-1$ positions in $S[p+1:p+m']$ for which there are no mutations at all. This occurs with probability at most $(1 - \theta_T)^{k-1}$. An overestimate gives that every point in $S[p+1:p+m']$ contributes a homologous anchor, showing that $\mathbb{E}(N_H) \le m'(1 - \theta_T)^{k-1} = \frac{m'}{1-\theta_T}(1-\theta_T)^k = O(m(1-\theta_T)^k)$ under $G$.
\end{proof}

Since $mk(1 - \theta_T)^k \gg m(1 - \theta_T)^k$, as $n \to \infty$, clipping anchors are, in fact, the dominant type of anchor.

\end{document}